\documentclass[runningheads,a4paper,orivec,envcountsame,envcountsect,final]{llncs}

\usepackage{ifdraft}
\ifdraft{
\usepackage[marginparwidth=2cm,nohead,nofoot,
            headsep = 1cm,
            footskip = 1cm,
            text={12.2cm,19.3cm},
            paperwidth=16.2cm,
            paperheight=23.3cm,
            marginparsep=1mm,
            marginparwidth=1.8cm,
            footskip=1cm,
            centering
            ]{geometry}
}{}

\hypersetup{
  bookmarks=true,
  colorlinks=true,
  linkcolor=black,
  citecolor=black,
  filecolor=black,
  urlcolor=black,
  pdftitle={On Well-Founded and Recursive Coalgebras},
  pdfauthor={Ji\v{r}\'\i\ Ad\'amek, Stefan Milius, and Lawrence S.~Moss},
  pdfkeywords={coalgebras, recursive, well-founded},
  pdfduplex={DuplexFlipLongEdge},
}
 
\usepackage[utf8]{inputenc}
\usepackage[T1]{fontenc}
\usepackage{etex,etoolbox}
\usepackage[british]{babel}

\makeatletter
\newenvironment{notheorembrackets}{%
\csdef{@spopargbegintheorem}##1##2##3##4##5{\trivlist%
      \item[\hskip\labelsep{##4##1\ ##2}]{##4{##3}\@thmcounterend\ }##5}%
    }{%
\csdef{@spopargbegintheorem}##1##2##3##4##5{\trivlist%
      \item[\hskip\labelsep{##4##1\ ##2}]{##4(##3)\@thmcounterend\ }##5}%
    }
\makeatother

\usepackage[%
rm={oldstyle=false,proportional=true},%
sf={oldstyle=false,proportional=true},%
tt={oldstyle=false,proportional=true,variable=true},%
qt=false%
]{cfr-lm}

\usepackage[noadjust]{cite} 

\addto\extrasbritish{
}

\ifdraft{%
  \makeatletter
  \setcounter{tocdepth}{2}
  \makeatother
}{}

\makeatletter
\renewcommand\subsection[1]{\@startsection{subsection}{2}{\z@}%
                       {-5\p@ \@plus -2\p@ \@minus -2\p@}%
                       {-0\p@ \@plus 0\p@ \@minus 0\p@}%
                       {\normalfont\normalsize\bfseries\boldmath}{#1}{\normalfont\normalsize\bfseries.\hspace{2mm}}}

\makeatother
                    
\usepackage[footnote,marginclue,nomargin]{fixme}
\FXRegisterAuthor{sm}{asm}{SM}
\FXRegisterAuthor{lm}{alm}{LM}
\FXRegisterAuthor{ja}{aja}{JA}

\newcommand{\smooth}{smooth\xspace}

\usepackage{tikz}

\newcommand{\takeout}[1]{\empty}

\usepackage{tikz}
\usetikzlibrary{cd,calc}

\newsavebox{\mypullbackcorner}%
\sbox{\mypullbackcorner}{%
\begin{tikzpicture}
    \draw[-] (0,0) -- (.5em,.5em) -- (0,1em);
\end{tikzpicture}%
}
\newcommand{\pullbackangle}[2][]{\arrow[phantom,to path={
                     -- ($ (\tikztostart)!1cm!#2:([xshift=8cm]\tikztostart) $)
                        node[anchor=west,pos=0.0,rotate=#2,
                        inner xsep = 0]
                        {\begin{tikzpicture}[minimum
                        height=1mm,baseline=0,#1]
    \draw[-] (0,0) -- (.5em,.5em) -- (0,1em);
                        \end{tikzpicture}}}]{}}

\tikzstyle{shiftarr}=[
        rounded corners,%
        to path={--([#1]\tikztostart.center)
                     -- ([#1]\tikztotarget.center) \tikztonodes
                     -- (\tikztotarget)},
]

                   \newsavebox{\kleisliarrow}
\savebox{\kleisliarrow}{%
\begin{tikzpicture}[
      baseline=(arrow.base),
      inner sep=8mm,
      outer sep=0mm,
      ]
      \node[draw=none,
      anchor=base,
      inner sep=0,
      outer sep=0,
      ] (arrow) {$\longrightarrow$};
    \draw[fill=white] ($ (arrow.south) !.68! (arrow.north)$) circle (0.15em);
  \end{tikzpicture}}

\newsavebox{\kleislidot}
\savebox{\kleislidot}{%
\begin{tikzpicture}[baseline=0pt,outer sep=0pt]
    \draw[fill=white] (0,0) circle (2pt);
  \end{tikzpicture}}

\tikzstyle{kleisli}=[
"{\usebox{\kleislidot}}"{red,anchor=center,font=\normalsize,pos=0.5},
outer sep = 1pt, 
]

\usepackage[all]{xy}    

\def\arKl[#1]{\ar[#1]|-{\circ}}
\SelectTips{cm}{}
\usepackage{microtype}
\usepackage{graphicx}
\usepackage{amsbsy,eucal}  
\usepackage{dsfont}
\usepackage{xspace}

\usepackage[inline]{enumitem}
\setlist[enumerate,1]{label=(\arabic*),font=\normalfont,align=left,leftmargin=0pt,labelindent=0pt,listparindent=\parindent,labelwidth=0pt,itemindent=!,topsep=3pt,parsep=0pt,itemsep=2pt,start=1}
\setlist[enumerate,2]{label=(\alph*),font=\normalfont,labelindent=*,leftmargin=*,start=1}
\setlist[itemize]{labelindent=*,leftmargin=*}
\setlist[description]{labelindent=*,leftmargin=*,itemindent=-1 em}

\usepackage{amsmath,amssymb,amscd,bbold,latexsym}
\usepackage{mathrsfs}
\usepackage{stmaryrd}

\numberwithin{equation}{section}

\spnewtheorem{assumption}[theorem]{Assumption}{\bfseries}{\rmfamily}
\spnewtheorem{notation}[theorem]{Notation}{\bfseries}{\rmfamily}
\spnewtheorem{observation}[theorem]{Observation}{\bfseries}{\rmfamily}
\spnewtheorem{defn}[theorem]{Definition}{\bfseries}{\rmfamily}
\spnewtheorem{expl}[theorem]{Example}{\bfseries}{\rmfamily}
\spnewtheorem{rem}[theorem]{Remark}{\bfseries}{\rmfamily}
\spnewtheorem{construction}[theorem]{Construction}{\bfseries}{\rmfamily}
\spnewtheorem{examples}[theorem]{Examples}{\bfseries}{\rmfamily}
\spnewtheorem*{proofsketch}{Proof sketch}{\itshape}{\rmfamily}
\spnewtheorem*{oproblem}{Open Problem}{\bfseries}{\rmfamily}

\newdir{ >}{{}*!/-5pt/\dir{>}}
\newdir{ (}{{}*!/-5pt/\dir^{(}}

\newcommand{\coker}{\ensuremath{\mathop{\mathrm{coker}}}}

\renewcommand{\phi}{\varphi}

\renewcommand{\o}{\cdot}

\def\id{{\mathit{id}}}
\def\Id{{\mathit{Id}}}

\def\inl{{\textsf{inl}}}
\def\inr{{\textsf{inr}}}

\def\andd{\wedge}
\def\orr{\vee}

\newlength{\mathfrwidth}
\newsavebox{\mathfrbox}

\newcommand{\ini}{\iota}

\renewcommand\vec[1]{\overrightarrow{#1}}
\newcommand\cev[1]{\overleftarrow{#1}}
\newcommand{\pback}[1]{\cev{#1}}
\newcommand{\pbackleft}[1]{\vec{#1}}
\newcommand{\pbackright}[1]{{#1}_*}

\def\A{\cat{A}}

\def\C{\cat{C}}

\def\Set{{\mathsf{Set}}}

\def\Ord{{\mathsf{Ord}}}
\def\KVec{\mathsf{Vec}_K}

\def\Coalg{\mathop{\mathsf{Coalg}}}

\def\colim{\mathop{\mathsf{colim}}}

\def\pow{{\mathscr P}}

\def\P{\pow}

\newcommand{\Pow}{\pow}

\newcommand{\pair}[1]{\langle #1 \rangle}

\def\epito{\twoheadrightarrow}
\def\monoto{\rightarrowtail}
\def\subto{\hookrightarrow}

\newcommand{\set}[1]{\{ #1\}}

\def\cat#1{{\mathscr{#1}}}

\newcommand{\pullbackcorner}[1][dr]{\save*!/#1-1pc/#1:(-1,1)@^{|-}\restore}
\newcommand{\nexttime}{\bigcirc}

\DeclareMathOperator{\Sub}{\textnormal{\textsf{Sub}}}

\def\eps{\varepsilon}

\def\Nat{\mathds{N}}

\def\Int{\mathds{Z}}

\def\Real{\mathds{R}}

\def\Reals{\Real}

\def\sol#1{{#1}^{\dag}}

\def\CPO{\ensuremath{\mathsf{CPO}}}

\def\Gra{\ensuremath{\mathsf{Gra}}}

\def\2{\textbf{2}}

\def\qsort{q}%
\def\qsplit{s}%
\def\qmerge{m}%

\usepackage[inline]{showlabels}
\showlabels{bibitem}
\usepackage{seqsplit}
\usepackage{xstring}
\usepackage{xcolor}



\title{On Well-Founded and Recursive Coalgebras}
\titlerunning{On Well-Founded and Recursive Coalgebras}
\author{Ji\v{r}\'\i\ Ad\'amek\inst{1}\fnmsep\thanks{Supported by the Grant Agency of the Czech Republic under grant 19-00902S.}
  \and
  Stefan Milius\inst{2}\fnmsep\thanks{Supported by Deutsche Forschungsgemeinschaft
    (DFG) under project MI~717/5-2}
  \and
  Lawrence S.~Moss\inst{3}\fnmsep\thanks{Supported by grant $\#$586136 from the Simons Foundation.}}
\authorrunning{J.~Ad\'amek, S.~Milius, and L.S.~Moss}

\institute{Czech Technical University, Prague,  Czech Republic \\
  \email{j.adamek@tu-braunschweig.de}
  \and
  Friedrich-Alexander-Universit\"at Erlangen-N\"urnberg, Germany\\
  \email{mail@stefan-milius.eu}
  \and
  Indiana University, Bloomington, IN, USA \\
  \email{lmoss@indiana.edu}}

\begin{document}
%
%
\maketitle
\begin{abstract}
    This paper studies fundamental questions concerning category-theoretic models of induction and recursion. We are concerned with the relationship between well-founded and recursive coalgebras for an endofunctor. For monomorphism preserving endofunctors on complete and well-powered categories every coalgebra has a well-founded part, and we provide a new, shorter proof that this is the coreflection in the category of all well-founded coalgebras. We present a new more general proof of Taylor's General Recursion Theorem that every well-founded coalgebra is recursive, and we study under which hypothesis the converse holds. In addition, we present a new equivalent characterization of well-foundedness: a coalgebra is well-founded iff it admits a coalgebra-to-algebra morphism to the initial algebra.
\end{abstract}

\section{Introduction}
\label{sec-wfc}

What is induction?  What is recursion?  In areas of theoretical
computer science, the most common answers are related to \emph{initial
  algebras}.  Indeed, the dominant trend in abstract data types is
initial algebra semantics (see e.g.~\cite{MG82}), and this approach
has spread to other semantically-inclined areas of the subject. The
approach in broad slogans is that, for an endofunctor $F$ describing
the type of algebraic operations of interest, the initial algebra
$\mu F$ has the property that for every $F$-algebra $A$, there is a
unique homomorphism $\mu F \to A$, and this \emph{is}
recursion. Perhaps the primary example is \emph{recursion on $\Nat$,
  the natural numbers}. Recall that $\Nat$ is the initial algebra for
the set functor $FX = X + 1$. If $A$ is any set, and $a\in A$ and
$\alpha\colon A \to A$ are given, then initiality tells us that there is a
unique $f\colon \Nat\to A$ such that for all $n\in \Nat$,
\begin{equation}\label{Nrec}
  f(0) =  a \qquad f(n+1) = \alpha(f(n)).
\end{equation}
Then the first additional problem coming with this approach is that of
how to ``recognize'' initial algebras: Given an algebra, how do we
really know if it is initial? The answer -- again in slogans -- is
that initial algebras are the ones with ``no junk and no confusion.''

Although initiality captures some important aspects of recursion, 
it cannot be a fully satisfactory approach.  One big missing piece 
concerns recursive definitions based on well-founded relations.
%
%
For example, the whole study of termination of rewriting systems
depends on well-orders, the primary example of \emph{recursion on a
  well-founded order}.  Let $(X, R)$ be a well-founded relation, i.e.~one
with no infinite sequences $\cdots x_2 \mathbin{R} x_1 \mathbin{R} x_0$.
Let $A$ be any set, and let $\alpha\colon \pow A \to A$.  (Here and below,
$\pow$ is the power set functor, taking a set to the set of its
subsets.)  Then there is a unique $f\colon X\to A$ such that for all
$x \in X$
\begin{equation}\label{Tf}
  f(x) = \alpha (\set{f(y) : y\ R\ x}). 
\end{equation}
The main goal of this paper is the study of concepts that allow to
extend the algebraic spirit behind initiality in \eqref{Nrec} to the
setting of recursion arising from well-foundedness as we find it in
\eqref{Tf}. The corresponding concepts are those of well-founded and
recursive coalgebras for an endofunctor, which first appear in work by
Osius~\cite{osius} and Taylor~\cite{taylor2,taylor3}, respectively.
In his work on categorical set theory, Osius~\cite{osius}
first studied the notions of well-founded and recursive coalgebras
(for the power-set functor on sets and, more generally, the
power-object functor on an elementary topos). He defined recursive
coalgebras as those coalgebras $\alpha\colon A  \to \pow A$ which have a unique
coalgebra-to-algebra homomorphism into every algebra (see \autoref{def:recoalg}).
\smnote{He \emph{did} kind of have ``recursive $\iff$ well-founded'';
  recursive $\Rightarrow$ well-founded under an additional assumption
  on the coalgebra in his Thm.~6.4.}

Taylor~\cite{taylor2,taylor3} took Osius' ideas much further. He
introduced well-founded coalgebras for a general endofunctor,
capturing the notion of a well-founded relation categorically, and
considered recursive coalgebras under the name `coalgebras obeying the
recursion scheme'.  He then proved the General Recursion Theorem that
all well-founded coalgebras are recursive for every endofunctor on
sets (and on more general categories)\smnote{He clearly has a list of
  properties in~\cite{taylor3}; so this is not just sets!}  preserving
inverse images. Recursive coalgebras were also investigated by
Eppendahl~\cite{eppendahl99}, who called them algebra-initial
coalgebras.  Capretta, Uustalu, and Vene~\cite{cuv06} further studied
recursive coalgebras, and they showed how to construct new ones from
given ones by using comonads. They also explained nicely how recursive
coalgebras allow for the semantic treatment of (functional)
divide-and-conquer programs. More recently, Jeannin et
al.~\cite{JeanninEA17} proved the general recursion theorem for
polynomial functors on the category of many-sorted sets; they also
provide many interesting examples of recursive coalgebras arising in
programming.

Our contributions in this paper are as follows. We start by recalling
some preliminaries in \autoref{S:prelim} and the definition of
(parametrically) recursive coalgebras in \autoref{S:reco} and of
well-founded coalgebras in \autoref{S:wfd} (using a formulation based
on Jacobs’ next time operator~\cite{Jacobs02}, which we extend from
Kripke polynomial set functors to arbitrary functors). We show that
every coalgebra for a monomorphism-preserving functor on a complete
and well-powered category has a well-founded part, and provide a new
proof that this is the coreflection in the category of well-founded
coalgebras (\autoref{P:wfdpart2}), shortening our previous
proof~\cite{amms}. Next we provide a new proof of Taylor's General
Recursion Theorem (\autoref{T:wf-prec}), generalizing this to
endofunctors preserving monomorphisms on a complete and well-powered
category having \smooth monomorphisms (see \autoref{D:constr}).  For
the category of sets, this implies that ``well-founded $\Rightarrow$
recursive'' holds for all endofunctors, strengthening Taylor's
result. We then discuss the converse: is every recursive coalgebra
well-founded? Here the assumption that $F$ preserves inverse images
cannot be lifted, and one needs additional assumptions. In fact, we
present two proofs: one assumes the functor has a pre-fixed point and
universally \smooth monomorphisms (see \autoref{T:rec-wf:1}). Under
these assumptions we also give a new equivalent characterization of
recursiveness and well-foundedness: a coalgebra is recursive if it has
a coalgebra-to-algebra morphism into the initial algebra (which exists
under our assumptions), see \autoref{C:equiv}. This characterization
was previously established for finitary functors on
sets~\cite{alm_rec}. The other proof of the above implication is due
to Taylor~\cite{taylor3} and presented for the convenience of the
reader. Taylor's proof uses the concept of a subobject classifier
(\autoref{T:rec-wf:2}). It implies that `recursive' and `well-founded'
are equivalent concepts for all set functors preserving inverse
images.  We also prove that a similar result holds for the category of
vector spaces over a fixed field (\autoref{C:vec}).

Finally, we show in \autoref{S:closure} that well-founded coalgebras
are closed under coproducts, quotients and, assuming mild assumptions,
under subcoalgebras.

\takeout{
Besides these technical contributions, we believe that our paper (or
rather its full version) provides the most comprehensive study to date
of results on well-founded and recursive coalgebras that may serve as
a reference for the subject.}


\takeout{
What is induction? What is recursion? One often encounters
explanations of these concepts based on the universal property of
initial algebras. However, from the work by Osius~\cite{osius} and
Taylor~\cite{taylor2,taylor3} a different approach appears: induction
and recursion are properties of coalgebras rather than of
algebras. For example, a set equipped with a well-founded binary
relation is equivalently a well-founded coalgebra for the power-set
functor. In his work on categorical set theory, Osius~\cite{osius}
first studied the notions of well-founded and recursive coalgebras
(for the power-set functor on sets and, more generally, the
power-object functor on an elementary topos). He defined recursive
coalgebras as those coalgebras $A$ which have a unique
coalgebra-to-algebra homomorphism into every algebra, and he proved the general
recursion theorem, that well-founded and recursive coalgebras are the
same.\smnote{He had recursive $\Rightarrow$
  well-founded under an additional assumption on the coalgebra in his Thm.~6.4.}

Taylor~\cite{taylor2,taylor3} took Osius' ideas much further. He
introduced well-founded coalgebras for a general endofunctor capturing
the notion of a well-founded relation categorically and considered
recursive coalgebras under the name `coalgebras obeying the recursion
sheme'.  He then proved the General Recursion Theorem for every
functor on sets (and on more general categories)\smnote{He clearly has
  a list of properties in~\cite{taylor3}; so this is not just sets!}
preserving inverse images.  Recursive coalgebras were also
investigated by Eppendahl~\cite{eppendahl99}, who called them
algebra-initial coalgebras.

Capretta, Uustalu, and Vene~\cite{cuv06} further studied recursive
coalgebras, and they showed how to construct new ones from given ones
by using comonads. They also explained nicely how recursive coalgebras
allow for the semantic treatment of (functional) divide-and-conquer
programs.

\takeout{
In our opinion, recursive coalgebras are the ``gist'' of induction: in
order to define a (homo-)morphism from~$A$ to~$X$, one only needs an
algebra structure on~$X$. For example, every initial algebra can be
considered as a coalgebra (recall that by Lambek's Lemma its algebra
structure is invertible), and it is recursive. But there are other
natural coalgebras used in induction and recursion, e.g.~Capretta et
al.~\cite{cuv06} provide nice examples, some of which are elaborated
below. Besides, there are interesting examples of recursive coalgebras
in situations where no initial algebra exists.}

In general, coalgebras are generic state-based systems whose type of
transitions is described by an endofunctor $F$. More precisely, a
coalgebra is a pair formed by an object $A$ (of states) and a morphism
$\alpha\colon A \to FA$ (modelling transisitons). For example, a
deterministic automaton over input alphabet $\Sigma$ on a state set
$A$ has final states and transitions expressed by a function
$\alpha\colon X \to \set{0,1} \times X^\Sigma$. Thus, here we use
$FX = \set{0,1} \times X^\Sigma$. Graphs are  another important source of examples:
 in order to give a graph on a set $A$ of vertices one assigns
to each vertex the sets $\alpha(a)$ of its neighbours, i.e.~a graph is
a coalgebra for the power-set functor $\Pow$ (defined by
$\Pow X = \set{M : M \subseteq X}$). Furthermore, labelled transition
system over the label alphabet $\Sigma$ are coalgebras
$\alpha\colon A \to \Pow(\Sigma \times A)$, thus here one uses the
functor $FX = \Pow(\Sigma \times X)$. Many other types of state-based
systems have such a representation, see e.g.~Rutten~\cite{rutten}. In
addition one may want to use other base categories besides sets if states
carry additional structure preserved by transitions, e.g.~linear
weighted automata over the input alphabet $\Sigma$ with weights from a
field $K$, which are given by a vector space $A$ and a linear map
$\alpha\colon A \to K \times A^\Sigma$ are coalgebras for
$FX = K \times X^\Sigma$ on the category of vector spaces over $K$.

We begin by giving a more compact definition of well-founded
coalgebras based on Jacobs' next time operator~\cite{Jacobs02}. He
defined this operator for so-called Kripke polynomial functors on
sets; intuitively it is the semantic counterpart of the well-known
next time modality in temporal logic (see e.g.~Manna and
Pn\"ueli~\cite{MP92}). We generalize the next time operator to
arbitrary functors, and one then easily sees that a coalgebra is
well-founded iff it has no proper subobject as a fixed point of the
next time operator.

The relevance of the General Recursion Theorem stems from the fact
that verifying whether a given coalgebra is recursive can be much more
difficult than verifying that it is well-founded. We generalize this
result and prove that for every endofunctor preserving monomorphisms on
a complete and well-powered category having \smooth monomorphisms
(see~\cite{takr}), every well-founded coalgebra is recursive
(\autoref{T:wf-prec}). For sets, this implies that the implication
\[
  \text{well-founded} \implies \text{recursive}
\]
hold for all endofunctors, which strengthens Taylor's result in the
case of sets.

We also discuss the converse: is every recursive coalgebra
well-founded? Here the assumption that $F$ preserves inverse images
cannot be lifted, and one needs additional assumptions.
We then presents two proofs for the
implication
\[
  \text{recursive} \implies \text{well-founded}.
\]
One assumes universally \smooth monomorphisms and that the
functor has a pre-fixed point (see \autoref{T:rec-wf:1}). Under these
assumptions we also give a new equivalent characterization of
recursiveness and well-foundedness: a coalgebra is recursive if it has
a coalgebra-to-algebra morphism into the initial algebra (which exists
under our assumptions), see \autoref{C:equiv}.  The other proof of the
above implication is an unpublished proof due to Taylor using the
concept of a subobject classifier (\autoref{T:rec-wf:2}). It implies
that `recursive' and `well-founded' are equivalent concepts for all
set functors preserving inverse images. This condition is actually
quite mild, e.g.~the above functors $\set{0,1} \times X^\Sigma$,
$\Pow$ and $\Pow(\Sigma \times X)$ do, and a composite, product or
coproduct of set functors preserving inverse images also preserves
them.

In our equivalence results, we also consider a stronger version of
recursivity called parametric recursivity, which is dual to Milius'
completely iterativity~\cite{milius}.

Finally, we show that all strong quotients, coproducts
(\autoref{C:wfcolim}) and subcoalgebras of well-founded coalgebras are
well-founded (\autoref{P:wfsub} and \autoref{T:sub}), and that every
coalgebra has a well-founded part, i.e.~the largest subcoalgebra which
is well-founded, and this is the coreflection of the given coalgebra
in the full subcategory of well-founded coalgebras
(\autoref{P:wfdpart2}).
}

\takeout{
Recursive coalgebras  have the 
property that functions out of them may be specified by structured
recursion. This is a desirable property, but it may not be
straightforward to establish in general. In this paper we consider
\emph{well-foundedness} of coalgebras, which captures well-founded
induction and is usually much easier to establish for a given
coalgebra. We will present a proof of Taylor's General
Recursion Theorem\smnote{I think it's better to mention Taylor already
  here and not only burried in the paper later.} that every
well-founded coalgebra is (parametrically) recursive. For set functors
preserving inverse images the converse holds, too. Moreover, if an initial
algebra exists, well-foundedness of a coalgebra $C$ is equivalent to
the existence of a coalgebra-to-algebra morphism from $C$ to the
initial algebra $\mu F$.\smnote{People want to use and quote this new
  result, so we should mention it upfront.}

Finally, we shall see that for every set functor, the initial algebra
is characterized as the terminal well-founded coalgebra. 
}

\section{Preliminaries}
\label{S:prelim}

We start by recalling some background material.  Except for the
definitions of \emph{algebra} and \emph{coalgebra} in
Section~\ref{S:algcoalg}, the subsections below may be read as
needed. We assume that readers are familiar with notions of basic
category theory; see e.g.~\cite{ahs} for everything which we do not detail.

\subsection{Algebras and Coalgebras}
\label{S:algcoalg}
We are concerned throughout this paper with \emph{algebras} and
\emph{coalgebras} for an endofunctor. This means that we have an
underlying category, usually written $\A$; frequently it is the
category of sets or of vector spaces over a fixed field, and that a functor
$F\colon\A \to \A$ is given.  An \emph{$F$-algebra} is a pair $(A,\alpha)$,
where $\alpha\colon FA \to A$.  An \emph{$F$-coalgebra} is a pair
$(A,\alpha)$, where $\alpha\colon A \to FA$. We usually drop the
functor $F$. Given two algebras $(A,\alpha)$ and $(B,\beta)$, an
\emph{algebra homomorphism} from the first to the second is
$h\colon A\to B$ in $\A$ such that the diagram below
commutes:
\[
  \begin{tikzcd}
    FA \arrow{r}{\alpha} \arrow{d}[swap]{Fh} & A \arrow{d}{h} \\
    FB \arrow{r}{\beta} & B
  \end{tikzcd}
\]
That is $h\o \alpha = \beta\o Fh$.  An algebra is \emph{initial} if it
has a unique morphism to every algebra. Recall that by Lambek's
Lemma~\cite{lambek}, whenever an initial algebra
$\ini\colon F(\mu F) \to \mu F$ exists, then $\ini$ is an
isomorphism. Thus, $\mu F$ can always be regarded as a coalgebra
$(\mu F, \iota^{-1})$.

Similarly, given coalgebras $(A,\alpha)$ and $(B,\beta)$, a
\emph{homomorphism of $F$-coalgebras} from the first to the second is
$h\colon A\to B$ in $\A$ such that $Fh\o \alpha = \beta\o h$.
Moreover, a terminal coalgebra is one with the property that every
coalgebra has a unique morphism into it. The category of $F$-coalgebras
is denoted by $\Coalg F$.
  
\begin{expl}\label{E:graph}
  \begin{enumerate}
  \item\label{E:graph:1} The power set functor
    $\pow\colon \Set\to\Set$ takes a set $X$ to the set $\pow X$ of
    all subsets of it; for a morphism $f\colon X\to Y$,
    $\pow f\colon \pow X\to \pow Y$ takes a subset $S\subseteq X$ to
    its direct image $f[S]$.  Coalgebras $\alpha\colon X \to \Pow X$
    may be identified with directed graphs on the set $X$ of vertices,
    and the coalgebra structure $\alpha$ describes the edges:
    $b \in \alpha(a)$ means that there is an edge $a \to b$ in the
    graph.
  
  \item Let $\Sigma$ be a signature, i.e.~a set of operation symbols,
    each with a finite arity. The \emph{polynomial functor} $H_\Sigma$
    associated to $\Sigma$ assigns to a set $X$ the set
    \[
      H_\Sigma X = \coprod_{n \in \Nat} \Sigma_n \times X^n,
    \]
    where $\Sigma_n$ is the set of operation symbols of arity $n$.  This
    may be identified with the set of all terms $\sigma(x_1,\ldots, x_n)$,
    for $\sigma\in \Sigma_n$, and $x_1, \ldots, x_n\in X$.
    Algebras for $H_\Sigma$ are the usual $\Sigma$-algebras.
  
  \item Deterministic automata over an input alphabet $\Sigma$ are
    coalgebras for the functor $FX = \set{0,1}\times X^\Sigma$.
    Indeed, given a set $S$ of states, the next-state map
    $S\times \Sigma\to S$ may be curried to
    $\delta\colon S\to S^\Sigma$. The set of final states yields the
    acceptance predicate $a\colon S\to \set{0,1}$. So the automaton
    may be regarded as
    $\pair{a,\delta}\colon S\to \set{0,1}\times S^\Sigma$.
    
  \item Labelled transitions systems are coalgebras for 
    $FX = \Pow(\Sigma \times X)$.
    
  \item To describe linear weighted automata, i.e.~weighted automata
    over the input alphabet $\Sigma$ with weights in a field $K$, as
    coalgebras, one works with the category $\KVec$ of vector spaces
    over $K$. A linear weighted automaton with the input alphabet
    $\Sigma$ is then a coalgebra for $FX = K \times X^\Sigma$.
  \end{enumerate}  
\end{expl}

\takeout{
\begin{expl}
\label{exFgraphs}
Let $\Gra$ be the category of graphs; the morphisms are map preserving
edges.  (This is not the same as the coalgebra category of $\pow$.)
Here is a simple endofunctor $F$ on $\Gra$ which we shall use at
various points.  On objects $A$ put $FA = 1$ (the terminal graph, it
has a loop) if $A$ has edges. For a graph $A$ without edges, let $FA$
be the graph without edges whose vertices are those of $A$ plus an
additional vertex; again, $FA$ has no edges.  The definition of $F$ on
morphisms $h\colon A \to B$ is as expected: $Fh$ maps the additional
vertex of $A$ to that of $B$ in the case where $B$ has no edges. Then
the initial algebra $\mu F$ is the graph of natural numbers without
edges.  The terminal coalgebra is $1\cong F1$.
\end{expl}
}

\begin{rem}\label{R:compwell}
  \begin{enumerate}
  \item\label{R:compwell:2} Recall that an epimorphism
    $e\colon A \to B$ is called \emph{strong} if it satisfies the
    following \emph{diagonal fill-in property}: given a monomorphism
    $m\colon C \monoto D$ and morphisms $f\colon A \to C$ and
    $g\colon B \to D$ such that $m\o f = g\o e$ (i.e.~the outside of
    the square below commutes) then there exists a unique
    $d\colon B \to C$ such that the diagram below commutes:
    \begin{equation}\label{diag:diag}
      \begin{tikzcd}
        A
        \arrow[->>]{r}{e}
        \arrow{d}[swap]{f}
        &
        B
        \arrow{d}{g}
        \arrow[dashed]{ld}[swap]{d}
        \\
        C
        \arrow[>->]{r}{m}
        &
        D
      \end{tikzcd} 
    \end{equation}

  \item\label{R:compwell:3} A complete and well-powered category $\A$
    has factorizations of morphisms $f$ as $f = m \cdot e$, where $e$
    is a strong epimorphism and $m$ is a monomorphism. It follows from
    Ad\'amek et al.~\cite[Theorem~14.17 and dual of
    Exercise~14C(d)]{ahs} that every complete and well-powered
    category has such factorizations.  We call the subobject $m$ the
    \emph{image} of $f$.

  \item We indicate monomorphisms by $\monoto$ and strong epimorphisms
    by $\epito$.
  \end{enumerate}
\end{rem}
\subsection{Preservation Properties}  
Recall that an intersection of two subobjects
$s_i\colon S_i \monoto A$ ($i = 1,2$) of a given object $A$ is given
by their pullback. Analogously, (general) intersections are given by
wide pullbacks.  Furthermore, the inverse image of a subobject
$s\colon S \monoto B$ under a morphism $f\colon A \to B$ is the
subobject $t\colon T \to A$ obtained by a pullback of $s$ along $f$.

\begin{expl}\label{E:setfunctors}
  The condition that a functor preserves intersections is an extremely
  mild one for set functors:  
  \begin{enumerate}
  \item Every polynomial functor preserves intersections and inverse
    images.
    
  \item The power-set functor $\Pow$ preserves intersections and
    inverse images.
    
  \item The collection of set functors which preserve intersections is
    closed under products, coproducts, and compositions. A subfunctor
    $m\colon G \monoto F$ of an intersection preserving functor $F$
    preserves intersections whenever $m$ is a cartesian natural
    transformation, i.e.~all naturality squares are pullbacks (being a
    pullback is indicated by the ``corner'' symbol):
    \[
      \begin{tikzcd}
        GX
        \pullbackangle{-45}
        \arrow[>->]{r}{m_X}
        \arrow{d}[swap]{Gf}
        &
        FX
        \arrow{d}{Ff}
        \\
        GY \arrow[>->]{r}{m_Y}
        &
        FY
      \end{tikzcd}
    \]
    Similarly, for inverse images.

  \item The functor $C_{01}\colon \Set \to \Set$ is defined by
    $C_{01}\emptyset = \emptyset$ and $C_{01}1 = 1$ for
    $X\neq\emptyset$. $C_{01}$ clearly preserves monomorphisms but it
    does not preserve finite intersections. Indeed, the empty
    intersection of $\set{0}, \set{1} \subto \set{0,1}$ is mapped to
    $\emptyset$; however the intersection of those subsets under
    $C_{01}$ is $1$, not $\emptyset$.
  
  \item\label{E:setfunctors:2} Consider next the set functor $R$
    defined by
    \(
    RX = \set{(x,y) \in X \times X \colon x \neq y} + \set{d}
    \)
    for sets $X$. For a function $f\colon X \to Y$ put
    \[
      Rf(d) = d \quad\text{and}\quad
      Rf(x,y) = \begin{cases}
        d & \text{if $f(x) \neq f(y)$}\\
        (f(x),f(y)) & \text{else.}
      \end{cases}
    \]
    This functor preserves finite intersections, since it
    preserves the above intersection of
    $\set{0}, \set{1} \subto \set{0,1}$, and so it is (naturally
    isomorphic to) its Trnkov\'a hull.
    \smnote{Jirka votes for deleting the argument, but I vote for
      keeping it.}
    However, $R$ does not preserve
    inverse images; consider e.g.~the pullback diagram (under $R$):
    \[
      \begin{tikzcd}[column sep = 10mm]
        \emptyset
        \arrow[hookrightarrow]{d}
        \arrow{r}
        \pullbackangle{-45}
        &
        \set{0}
        \arrow[hookrightarrow]{d}
        \\
        \set{0,1}
        \arrow{r}{\mathsf{const}_1}
        &
        \set{0,1}
      \end{tikzcd}
    \]
    (For $(0,1) \in R\set{0,1}$ and $d \in R\set{0}$ are merged in
    the right-hand $R\set{0,1}$, yet there is no
    suitable element in $R\emptyset$.)
    \smnote{TODO: @Jirka: what's an example of a finite
      intersection preserving functor not preserving intersections?}
    
  \item ``Almost'' all finitary set functors
    preserve intersections. In fact, the Trnkov\'a hull of a finitary
    set functor preserves intersections (see \autoref{P:Trint}).
  \end{enumerate}
\end{expl}

Some of our results require $F$ to preserve finite (or all)
intersection or inverse images. For set functors these are rather mild
requirements, as we now explain.

\begin{notheorembrackets}
\begin{proposition}[\cite{trnkova71}]\label{P:Tr}
  For every set functor $F$ there exists an essentially unique set
  functor $\bar F$ which coincides with $F$ on nonempty sets and
  functions and preserves finite intersections (whence monomorphisms).
\end{proposition}
\end{notheorembrackets}
\noindent
For the proof see Trnkov\'a~\cite[Propositions III.5 and II.4]{trnkova71}; for a
more direct proof see Ad\'amek and Trnkov\'a~\cite[Theorem~III.4.5]{ATbook}.
We call the functor $\bar F$ the \emph{Trnkov\'a hull} of $F$. 
\begin{rem}\label{R:inter}
  In fact, Trnkov\'a gave a construction of $\bar F$: she defined
  $\bar F \emptyset$ as the set of all natural transformations
  $C_{01} \to F$, where $C_{01}$ is the set functor with
  $C_{01} \emptyset = \emptyset$ and $C_{01} X =1$ for all nonempty
  sets $X$. For the empty map $e\colon \emptyset \to X$ with
  $X\neq \emptyset$, $\bar F e$ maps a natural transformation
  $\tau\colon C_{01} \to F$ to the element given by
  $\tau_X\colon 1\to FX$.
\end{rem}
Preservation of all intersections can be achieved for \emph{finitary}
set functors. Intuitively, a functor on sets is
finitary if its behavior is completely determined by its action on
\emph{finite} sets and functions. For a general functor, this
intuition is captured by requiring that the functor preserve filtered
colimits~\cite{ar}. For a set functor $F$ this is equivalent to being
\emph{finitely bounded}, which is the following condition: for each
element $x \in FX$ there exists a finite subset $M \subseteq X$ such
that $x \in Fi[FM]$, where $i\colon M \subto X$ is the inclusion
map~\cite[Rem.~3.14]{amsw19}.

\begin{notheorembrackets}
\begin{proposition}[{\cite[p.~66]{amm18}}]\label{P:Trint}
  The Trnkov\'a hull of a finitary set functor preserves all
  intersections.
\end{proposition}
\end{notheorembrackets}
\begin{proof}
  Let $F$ be a finitary set functor. Since $\bar F$ is finitary and
  preserves finite intersections, for every element $x \in \bar F X$,
  there exists a \emph{least} finite set $m\colon Y \subto X$ with $x$
  contained in $\bar F m$. Preservation of all intersections now
  follows easily: given subsets $v_i\colon V_i \subto X$, $i \in I$, with
  $x$ contained in the image of $\bar F v_i$ for each $i$, then $x$
  also lies in the image of the finite set $v_i \cap m$, hence
  $m \subseteq v_i$ by minimality. This proves
  $m \subseteq\bigcap_{i\in I} v_i$, thus, $x$ lies in the image of
  $\bar F(\bigcap_{i\in I} v_i)$, as required.\qed
\end{proof}

\subsection{Factorizations}
Every complete and well-powered category $\A$ has the following
factorizations of morphisms: every morphism $f$ may be written as
$f = m \cdot e$, where $e$ is a strong epimorphism and $m$ is a
monomorphism~\cite[Prop.~4.4.3]{Borceux94}.
We call the subobject $m$ the
\emph{image} of $f$. It follows from a result in Kurz'
thesis~\cite[Prop.~1.3.6]{Kurz00} that factorizations of morphisms
lift to coalgebras:

\begin{proposition}[$\Coalg F$ inherits factorizations from
  $\A$]\label{P:(e,m)}
  Suppose that $F$ preserves monomorphisms. Then the category
  $\Coalg F$ has factorizations of homomorphisms $f$ as
  $f = m \cdot e$, where $e$ is carried by a strong epimorphism and
  $m$ by a monomorphism in $\A$. The diagonal fill-in property holds
  in $\Coalg F$.
\end{proposition}

\begin{rem}\label{R:subcoalg}
  By a \emph{subcoalgebra} of a coalgebra $(A,\alpha)$ we mean a
  subobject in $\Coalg F$ represented by a homomorphism
  $m\colon (B,\beta) \monoto (A,\alpha)$, where $m$ is monic in
  $\A$. Similarly, by a \emph{strong quotient} of a coalgebra
  $(A,\alpha)$ is represented by a homomorphism
  $e\colon (A,\alpha) \epito (C,\gamma)$ with $e$ strongly epic in
  $\A$.
\end{rem}

\subsection{Subobject Lattices}
\begin{notation}
  For every object $A$ we denote by $\Sub(A)$ the poset of subobjects
  of $A$. The top of this poset is represented by $\id_A$, and the
  bottom $\bot_A$ is the intersection of all subobjects of $A$.
\end{notation}
Now suppose that  $\A$ is a complete and well-powered category.
\begin{rem}\label{R:subcomp}
  Note that $\Sub(A)$ is a complete lattice: it is small since $\A$ is
  well-powered, and a meet of subobjects $m_i\colon A_i \monoto A$,
  $i \in I$, is their intersection, obtained by forming their wide
  pullback. It follows that $\Sub(A)$ has all joins as well.
\end{rem}
We shall need that forming inverse images, i.e.~pulling back along a
morphism, is a right adjoint.
\begin{notation}
  For every morphism $f\colon B \to A$ we have two operators:
  \begin{enumerate}
  \item The inverse image operator
    \[
      \pback f\colon \Sub(A) \to \Sub(B),
    \]
    assigning to every subobject $s\colon S \monoto A$ its inverse image
    under $f$ obtained by the following pullback
    \[
      \begin{tikzcd}
        P
        \arrow{r}
        \arrow[>->]{d}[swap]{\pback f(s)}
        \pullbackangle{-45}
        &
        S\arrow[>->]{d}{s}
        \\
        B \arrow{r}{f} & A
      \end{tikzcd}
    \]
  \item The (direct) image operator
    \[
      \pbackleft f\colon \Sub(B) \to \Sub(A),
    \]
    assigning to every subobject $t\colon T \monoto B$ the image of $f
    \cdot t$:
    \[
      \begin{tikzcd}
        T
        \arrow[>->]{d}[swap]{t}
        \arrow[->>]{r}
        &
        S
        \arrow[>->]{d}{\pbackleft f(t)}
        \\
        B
        \arrow{r}{f}
        &
        A
      \end{tikzcd}
    \]
  \end{enumerate}
\end{notation}
\begin{rem}\label{R:adjoint}
  \begin{enumerate}
  \item A monotone map $r\colon X \to Y$ between posets, regarded as a
    functor from $X$ to $Y$ considered as categories, is a right
    adjoint iff there exists a monotone map $\ell\colon Y \to X$ such
    that
    \[
      \ell(y) \leq x
      \qquad\text{iff}\qquad
      y \leq r(x)\qquad\text{for every $x \in X$ and $y \in Y$}. 
    \]
    
  \item Moreover, a monotone map $r\colon \Sub(B) \to \Sub(A)$ is a
    right adjoint iff it preserves intersections.  Indeed, the
    necessity follows since right adjoints preserve limits. For the
    sufficiency, suppose that $r$ preserves intersections, and define
    $\ell\colon \Sub(A) \to \Sub(B)$ by
    \[
      \ell(m) = \bigwedge_{m \leq r(m)} m
      \qquad
      \text{for every $m \in \Sub(A)$}.
    \]
    Then $\ell$ is clearly monotone, and for every $m'$ in $\Sub(B)$ we have
    \[
      \ell(m) \leq m' \qquad\text{iff}\qquad m \leq r(m').
    \]
    Thus, $\ell$ is the desired left adjoint of $r$.
  \end{enumerate}
\end{rem}
\begin{proposition}\label{P:subadjs}
  If $\A$ is complete and well-powered, then for every morphism
  $f\colon B \to A$ we have an adjoint situation:
  \[
    \begin{tikzcd}
      \Sub(A) \arrow[yshift=-2.5mm]{r}[swap]{\pback f}
      \arrow[phantom]{r}[description]{\bot}
      &
      \Sub(B).
      \arrow[yshift=2mm]{l}[swap]{\pbackleft f}
    \end{tikzcd}
  \]
  In other words: $\pbackleft f(t) \leq s$ iff $t \leq \pback
  f(s)$ for all subobjects $s\colon S \monoto A$ and $t\colon T
  \monoto B$. 
\end{proposition}
\begin{proof}
  In order to see this we consider the following diagram:
  \[
    \begin{tikzcd}
      T
      \arrow[->>]{rrr}{e}
      \arrow[>->,dashed]{rd}
      \arrow[>->,rounded corners,to path=
      {|- (\tikztotarget)}
      ]{rdd}
      \arrow[phantom,xshift=-1.5mm]{dd}{\labelstyle t}
      &&&
      I
      \arrow[>->,dashed]{ld}
      \arrow[>->,rounded corners,to path=
      {|- (\tikztotarget)}
      ]{ldd}
      \arrow[phantom,xshift=4mm]{dd}{\labelstyle\pbackleft f(t)}
      \\
      &
      P
      \arrow{r}
      \arrow[>->]{d}[swap]{\pback f(s)}
      \pullbackangle{-45}
      &
      S
      \arrow[>->]{d}{s}
      \\
      \phantom{ }
      &
      B
      \arrow{r}{f}
      &
      A
      &
      \phantom{ }
    \end{tikzcd}
  \]
  By the universal property of the lower middle pullback square and
  the diagonal fill-in property, we have the dashed morphism on the
  left iff we have the one on the right. Thus, $t \leq \pback f(s)$
  iff $\pbackleft f(t) \leq s$, as desired. \qed
\end{proof}

\subsection{Chains}
By a \emph{transfinite chain} in a category $\A$ we understand a
functor from the ordered class $\Ord$ of all ordinals into
$\A$. Moreover, for an ordinal $\lambda$, a \emph{$\lambda$-chain} in
$\A$ is a functor from $\lambda$ to $\A$. A category \emph{has
  colimits of chains} if for every ordinal $\lambda$ it has a colimit
of every $\lambda$-chain. This includes the initial object $0$ (the
case $\lambda = 0$).

\begin{defn}\label{D:constr}
  \begin{enumerate}
  \item A category $\A$ has \emph{\smooth monomorphisms} if for every
    $\lambda$-chain $C$ of monomorphisms a colimit exists, its
    colimit cocone is formed by monomorphisms, and for every cone of $C$
    formed by monomorphisms, the factorizing morphism from $\colim C$
    is monic. In particuar, every morphism from $0$ is monic.
 
  \item\label{D:constr:2} $\A$ has \emph{universally \smooth
      monomorphisms} if $\A$ also has pullbacks, and for every morphism
    $f\colon X \to \colim C$, the functor $\A/\colim C \to \A/X$ forming
    pullbacks along $f$ preserves the colimit of $C$.  This implies that the initial
    object $0$ is \emph{strict}, i.e.~every morphism $f\colon X \to 0$
    is an isomorphism. Indeed, consider the empty chain ($\lambda =
    0$). 
  \end{enumerate}
\end{defn}
\begin{expl}\label{E:uconstr}
  \begin{enumerate}
  \item $\Set$ has universally \smooth monomorphisms. More
    generally, every Grothendieck topos does.\smnote{TODO: Why is this
      the case; do we have a proof or reference?}
    
  \item\label{E:uconstr:2} $\KVec$ has \smooth monomorphisms, but not universally so
    because the initial object is not strict.

  \item Categories in which colimits of chains and pullbacks are
    formed ``set-like'' have universally \smooth monomorphisms.
    These include the categories of posets, graphs, topological
    spaces, presheaf categories, and many varieties, such as monoids,
    graphs, and unary algebras.

  \item\label{E:uconstr:4} Every locally finitely presentable category
    $\A$ with a strict initial object has \smooth
    monomorphisms. This follows from~\cite[Prop.~1.62]{ar}. Moreover,
    since pullbacks commute with colimits of chains, it is easy to
    prove that colimits of chains are universal. Indeed, suppose that
    $c_i\colon C_i \to C$ is the colimit cocone of some chain of
    objects $C_i$, $i < \lambda$, and let $f\colon B \to C$ be a
    morphism. Form the pullback of every $c_i$ along $f$:
    \[
      \begin{tikzcd}
        B_i \arrow{d}[swap]{b_i}
        \pullbackangle{-45}
        \arrow{r}{f_i}
        &
        C_i
        \arrow{d}{c_i}
        \\
        B
        \arrow{r}{f}
        &
        C
      \end{tikzcd}
    \]
    Then $b_i\colon B_i \to B$ is a colimit cocone. Indeed, in the
    category of commmutative squares in $\A$, the chain of
    the above pullbacks squares has as a colimit the following pullback square
    \[
      \begin{tikzcd}
        \colim B_i \arrow{d}[swap]{\cong}
        \pullbackangle{-45}
        \arrow{r}{\colim f_i}
        &
        \colim C_i = C
        \arrow[equals]{d}
        \\
        B
        \arrow{r}{f}
        &
        C
      \end{tikzcd}
    \]

    Unfortunately, the example of rings demonstrates that the assumption
    of strictness of $0$ cannot be lifted. In fact, the collections of
    monomorphisms is not \smooth in the category of rings since there
    exist non-injective homomorphisms whose domain is the initial ring
    $\Int$.

    \smnote{Some text taken out here!}
\takeout{    
  \item 
  The category of pointed sets does not have    universally \smooth monomorphisms,
  since its initial object is not strict.
  An example of a variety which does not have
    universally \smooth monomorphisms is the category of
    rings: the initial ring is $\mathbb{Z}$, the integers, and it is
    not strict.
The same holds for monoids and for vector spaces over a fixed field.
}    
\takeout{  \item Monoids are another example of a variety which does not have
    universally \smooth monomorphisms.  Indeed, $\Sub(A)$ has continuous
    meets for every monoid $A$, but in general this is not a frame.
    For example, let $A$ be $\Nat$, the additive monoid of natural
    numbers, and let $k\Nat$ be the submonoid $\set{ki \colon i \in N}$.
    Then $2\Nat = (4\Nat \orr 6\Nat ) = 4\Nat \orr 6\Nat$.  This is
    different from
    \[
      (2\Nat \andd 6\Nat ) \orr (2\Nat \andd 6\Nat ) = 2\Nat.
    \]
  }

  \item The category $\CPO$ of complete partial orders (i.e.~partially
    ordered sets with joins of $\omega$-chains) does not have
    \smooth monomorphisms. Indeed, consider the $\omega$-chain of
    linearly ordered sets $A_n = \set{0,\ldots,n} + \set{\top}$
    (where $\top$ is a top element) with inclusion maps $A_n \to A_{n+1}$.  Its
    colimit is the linearly ordered set $\Nat + \set{\top,\top'}$ of
    natural numbers with two added top elements $\top' < \top$.  For
    the sub-cpo $\Nat + \set{\top}$, the inclusions of $A_n$ are monic
    and form a cocone. But the unique factorizing morphism from the
    colimit is not monic.
\end{enumerate}
\end{expl}
\begin{rem}\label{R:constr}
  If $\A$ is a complete and well-powered category, then $\Sub(A)$ is a
  complete lattice. Now suppose that $\A$ has \smooth monomorphisms.

  \begin{enumerate}
  \item\label{R:constr:1} In this setting, the unique morphism
    $\bot_A\colon 0 \to A$ is a monomorphism and therefore the bottom
    element of the poset $\Sub(A)$.
    
  \item\label{R:constr:2}  Furthermore, a join of a chain in $\Sub(A)$
    is obtained by forming a colimit. 
    More precisely, given an
    ordinal $k$ and an $k$-chain $m_i\colon A_i \monoto A$ of
    subobjects $(i < k)$, we have the diagram of objects
    $(A_i)_{i <k}$, where for all $i \leq j < k$ the connecting
    morphisms $a_{ij}\colon A_i \monoto A_j$ are the unique
    factorizations witnessing $m_i \leq m_j$:
    \[
      \begin{tikzcd}
        A_i \arrow{rr}{a_{ij}} \arrow{rd}[swap]{m_i} & & A_j
        \arrow{ld}{m_j} \\
        & A
      \end{tikzcd}
    \]
    The colimit $B$ of this diagram is formed by monomorphisms
    $b_i\colon A_i \to B$, $i < k$, and the unique monomorphism $m\colon B
    \monoto A$ with $m \cdot b_i = m_i$ for all $i <k$ is the
    join of all $m_i$, in symbols: $m = \bigvee_{i < k} m_i$.
        
  \item\label{R:constr:3} If $\A$ has universally \smooth monomorphisms, then
    for every morphism $f\colon A \to B$, the operator $\pback
    f\colon \Sub(B) \to \Sub(A)$ preserves unions of chains.  
    
    Indeed, suppose that $c\colon C \monoto A$ is the union of a chain
    of subobjects $a_i\colon A_i \monoto A$ in $\Sub(A)$. Then $C$ is
    the colimit of the (chain of connecting morphisms between the)
    $A_i$ with colimit injections $c_i\colon A_i\monoto C$, say.  For
    the morphism $p = \pback f (c)\colon P \monoto B$ we paste two
    pullback squares for every $i$ as shown below:
    \[
      \begin{tikzcd}
        B_i
        \arrow[>->]{d}[swap]{p_i}
        \arrow[>->]{r}{f_i}
        \pullbackangle{-45}
        \arrow[>->,shiftarr = {xshift=-20pt}]{dd}[swap]{b_i}
        &
        A_i
        \arrow[>->]{d}{c_i}
        \arrow[>->,shiftarr = {xshift=20pt}]{dd}{a_i}
        \\
        P
        \arrow[>->]{d}[swap]{p}
        \arrow[>->]{r}{f'}
        \pullbackangle{-45}
        &
        C
        \arrow[>->]{d}{c}
        \\
        B
        \arrow[>->]{r}{f}
        &
        A
      \end{tikzcd}
    \]
    The outside is then the pullback square stating that
    $b_i = \pback f (a_i)$. By universality, $P = \colim B_i$ with
    colimit injections $p_i$. Thus, by the constructivity of
    monomorphisms $p$ is the union of the subobjects $b_i$ in
    $\Sub(B)$; in symbols:
    $\bigvee_i \pback f (a_i) = \pback f\big(\bigvee_i a_i\big)$ as
    desired.
  \end{enumerate}
\end{rem}
\begin{rem}\label{R:ini}
  \begin{enumerate}
  \item Suppose that $\A$ has colimits of chains. Recall~\cite{A74}
    that every endofunctor $F\colon \A \to \A$ gives rise to an
    essentially unique chain $W\colon \Ord \to \A$, the
    \emph{initial-algebra chain}, of objects $W_i = F^i 0$,
    $i \in \Ord$ and connecting morphisms
    $w_{ij}\colon F^i 0 \to F^j 0$, $i \leq j \in \Ord$. They are
    defined by transfinite recursion:
    \[
      \begin{array}{l@{\,}l@{\ }ll}
        W_0 & = &  0, \\
        W_{j+1}  & = & FW_j & \mbox{ for all ordinals $j$,}\\
        W_j &=  & \colim_{i<j} W_i & \mbox{ for all limit ordinals $j$,}\\
      \end{array}
    \]
    and
    \[
      \begin{array}{l}
        w_{0,1}\colon 0\to  W_0 \mbox{ is unique },\\
        w_{j+1,k+1} = F w_{j,k} \colon FW_j \to FW_k,\\
        w_{i,j} \; (i<j) \mbox{ is the colimit cocone for limit ordinals } j.
      \end{array}
    \]
  \item\label{R:ini:2} Now suppose that $\A$ has \smooth monomorphisms and that
    $F\colon \A \to \A$ has a \emph{pre-fixed point}, i.e.~an object
    $A$ with a monomorphism $\alpha\colon FA \monoto A$. Then an
    initial algebra exists. This follows from
    results by Trnkov\'a et al.~\cite{takr} as we now briefly recall.
    Let $\alpha\colon FA \monoto A$ be a pre-fixed point.  Then there
    is a unique cocone $\alpha_i\colon W_i \to B$ satisfying
    $\alpha_{i+1} = \alpha \o F\alpha_i$. Moreover, each $\alpha_i$ is monomorphic.
    Since $A$ has only a set of subobjects, there is some $\lambda$
    such that for every $i > \lambda$, all of the morphisms $\alpha_i$ represent the
    same subobject of $A$. Consequently, $w_{\lambda,\lambda +1}$ is
    an isomorphism. Then $\mu F = F^\lambda 0$ with the structure
    $\ini = w^{-1}_{\lambda,\lambda +1}\colon F(\mu F) \to \mu F$ is an
    initial algebra.
  \end{enumerate}
\end{rem}

\takeout{
\smnote[inline]{Remnants of old prelims follow.}
\begin{expl}\label{E:graph}
  Coalgebras $\alpha\colon S \to \Pow S$ for the power-set functor may
  be identified with directed graphs on the set $S$ of vertices, and
  the coalgebra structure $\alpha$ describes the edges: $b \in
  \alpha(a)$ means that the is an edge $a \to b$ in the graph.
  
  Occasionally, it is more convenient to reverse the direction of
  edges: the relation $R \subseteq S \times S$ given by $x \mathbin{R}
  y$ iff $x \in \alpha(y)$ is called the \emph{derived relation} of
  the coalgebra $(S,\alpha)$. 
\end{expl}
}%

\section{Recursive Coalgebras}
\label{S:reco}

\begin{assumption}\label{A:wfd}
  We work with a standard set theory (e.g. Zermelo-Fraenkel), assuming
  the Axiom of Choice. In particular, we use transfinite induction on
  several occasions. (We are not concerned with constructive
  foundations in this paper.)

  Throughout this paper we assume that $\A$ is a complete and
  well-powered category $\A$ and that $F\colon \A\to \A$ preserves
  monomorphisms.
\end{assumption}
\noindent
For $\A = \Set$ the condition that $F$ preserves
monomorphisms may be dropped. In fact, preservation of nonempty monomorphism
is sufficient in general (for a suitable notion of nonempty
monomorphism)~\cite[Lemma~2.5]{mpw19}, and this holds for every set functor.

The following definition of recursive coalgebras was first given by
Osius~\cite{osius}. Taylor~\cite{taylor2} speaks of
\emph{coalgebras obeying the recursion scheme}. Capretta et
al.~\cite{cuv06} extended the concept to \emph{parametrically
  recursive} coalgebra by dualizing completely iterative
algebras~\cite{milius}.
\begin{defn}\label{def:recoalg}
  A coalgebra $\gamma\colon C \to FC$ is called \emph{recursive} if for
  every algebra $\alpha\colon FA \to A$ there exists a unique
  coalgebra-to-algebra morphism $h\colon C \to A$, i.e.~a unique morphism
  such that the square below commutes:
  \[
    \begin{tikzcd}
      C \arrow{d}[swap]{\gamma} \arrow{r}{h}
      &
      A
      \\
      FC
      \arrow{r}{Fh}
      &
      FA
      \arrow{u}[swap]{\alpha}
    \end{tikzcd}
  \]
\end{defn}
 
\begin{examples}\label{E:reco}
  \begin{enumerate}
  \item The first examples of recursive coalgebras are well-founded
    relations. Recall that a binary relation $R$ on a set
    $X$ is well-founded if there is no infinite descending
    sequence
    \[
      \cdots \mathbin{R} x_3 \mathbin{R} x_2 \mathbin{R} x_1 \mathbin{R} x_0.
    \]
    Now a binary relation $R \subseteq X \times X$ is essentially a
    graph on $X$, equivalently the coalgebra structure
    $\alpha\colon X \to \pow X$ with
    $\alpha(x) = \set{y \mid y \mathbin{R} x}$
    (cf.~\autoref{E:graph}\ref{E:graph:1}). Osius~\cite{osius} showed
    that for every well-founded relation the associated
    $\pow$-coalgebra is recursive. Shortly: a graph regarded as a
    coalgebra for $\Pow$ is recursive iff it has no infinite path.
    
  \item If $\mu F$ exists, then it is a recursive coalgebra. 
    
  \item\label{E:reco:3} The initial coalgebra $0 \to F0$ is
    recursive.
    
  \item If $(C, \gamma)$ is recursive so is $(FC, F\gamma)$,
    see~\cite[Prop.~6]{cuv06}.
    
  \item\label{E:reco:5} Every colimit of recursive coalgebras in
    $\Coalg F$ is recursive.  This is easy to prove, using that
    colimits of coalgebras are formed on the level of the underlying
    category.

  \item\label{E:reco:6} It follows from
    items~\ref{E:reco:3}--\ref{E:reco:5} that in the initial-algebra
    chain from \autoref{R:ini} all coalgebras
    $w_{i,i+1}\colon F^i 0 \to F(F^{i} 0)$, $i \in \Ord$, are
    recursive.
  \end{enumerate}
\end{examples}

By an argument similar to the proof of the (dual of) Lambek's Lemma,
we see that a terminal recursive $F$-coalgebra is a fixed point of
$F$, and we have
\begin{notheorembrackets}%
\begin{corollary}[{\cite[Prop.~7]{cuv06}}]\label{cor:cuv}
  The initial algebra is precisely the same as the terminal
  recursive coalgebra.
\end{corollary}
\end{notheorembrackets}

Capretta et al.~\cite{cuv06} study the notion of a parametrically
recursive coalgebra dualizing the notion of a completely iterative
algebra~\cite{milius}.

\begin{defn}
  A coalgebra $(A,\alpha)$ is \emph{parametrically
    recursive} if for every morphism $e\colon FX \times A \to X $ there is a
  unique morphism $e^\dag\colon A\to X$ so that the square below
  commutes:
  \begin{equation}\label{eq:prec}
    \begin{tikzcd}[column sep = 15mm]
      A \arrow{r}{e^\dag} \arrow{d}[swap]{\pair{\alpha,A}}
      &
      X
      \\
      FA \times A
      \arrow{r}{Fe^\dag \times A}
      &
      FX \times A \arrow{u}[swap]{e}
    \end{tikzcd}
  \end{equation}
\end{defn}

The dual statement of~\cite[Thm.~2.8]{milius} states that the initial
algebra is, equivalently, the terminal parametrically recursive
coalgebra. Of course, every parametrically recursive coalgebra is
recursive. (To see this, form for a given $e\colon FX \to X$ the
morphism $e' = e\o \pi$, where $\pi\colon FX\times A\to FX$ is the
projection.) In Corollaries~\ref{C:equiv} and~\ref{C:equiv:2} we will
see that the converse often holds. However, in general the converse fails:
\begin{notheorembrackets}%
\begin{expl}[\cite{AdamekLuckeMilius07}]\label{E:functorR}
  Let $R\colon \Set \to \Set$ be the functor defined in
  Example~\ref{E:setfunctors}, part~\ref{E:setfunctors:2}.  Also, let
  $C = \{0,1\}$, and define $\gamma\colon C \to RC$ by
  $\gamma(0) = \gamma(1) = (0,1)$. Then $(C,\gamma)$ is a recursive
  coalgebra. Indeed, for every algebra $\alpha\colon RA \to A$ the
  constant map $h\colon C \to A$ with $h(0) = h(1) = \alpha(d)$ is the
  unique coalgebra-to-algebra morphism.

  However, $(C,\gamma)$ is not parametrically recursive. To see this,
  consider any morphism $e\colon RX \times \{0,1\} \to X$ such that $RX$
  contains more than one pair $(x_0,x_1)$, $x_0 \neq x_1$ with
  $e((x_0,x_1),i) = x_i$ for $i = 0,1$. Then each such pair yields
  $h\colon C \to X$ with $h(i) = x_i$ making~\eqref{eq:prec}
  commute. Thus, $(C,\gamma)$ is not parametrically recursive.
\end{expl}
\end{notheorembrackets}

The situation in \autoref{E:functorR} is relatively rare and artificial
because for functors preserving inverse images, recursive and
parametrically recursive coalgebras coincide (see \autoref{C:equiv}
and \autoref{C:equiv:2}).

\takeout{
The following example is an easy special case of the argument
given in Theorem~\ref{T:rec-wf:2} and so might serve as an intuitive
explanation for the latter.

\begin{expl}
  Let $\gamma\colon G \to \pow G$ be a graph, considered as a coalgebra for
  the power set functor $\pow$ on $\Set$. Then $\gamma$ is recursive
  iff $G$ has no infinite paths.

  Here is the reason: Consider the following $\pow$-algebra
  $(2,\tau)$, where $2 = \set{t,f}$, and $\tau\colon \pow 2 \to 2$ is given
  by $\tau(x) = f$ if $x = \set{t}$ or $x = \emptyset$, and
  $\tau(x) = f$ otherwise.  We have two coalgebra-to-algebra
  morphisms.  One is $h_1\colon G \to 2$, where $h_1$ is the constant
  function with value $t$.  The second is $h_2\colon G \to 2$, where
  $h_2(g) = t$ iff there is no infinite path in $G$ starting from $g$.
  It is not hard to check that $h_2$ is a coalgebra-to-algebra
  morphism.  It follows that $\gamma$ is recursive iff $h_1 = h_2$,
  and this holds iff there are no infinite paths in the graph $G$.
\end{expl}}

\takeout{
\begin{expl}
  A graph regarded as a coalgebra for $\Pow$ is recursive iff it has
  no infinite path. This follows from
  \autoref{E-well-founded}\ref{E-well-founded:0} and
  \autoref{C:equiv:2} below.
\end{expl}}

We conclude this section with a few examples explaining how recursive coalgebras
capture familiar recursive function definitions as well as functional
divide-and-conquer programs.
\begin{examples}\label{E:prec}
  \begin{enumerate}
  \item\label{E:prec:1} The functor $FX = X + 1$ has unary algebras
    with a constant as algebras, and coalgebras for $F$ may be
    identified with partial unary algebras. The initial algebra for
    $F$ is the set of natural numbers $\Nat$ with the structure given
    by the successor function and the constant $0$. The inverse of
    the initial $F$-algebra is the coalgebra given by the partial
    unary operation $n \mapsto n-1$ (defined iff $n > 0$). This
    coalgebra is parametrically recursive. Hence every function
    \[
      e = [u,v]\colon F X \times \Nat \cong  \Nat + X \times \Nat  \to X
    \]
    defines a unique sequence $\sol e\colon \Nat \to X$, $\sol e(n) = x_n$
    such that~\eqref{eq:prec} commutes. This means that $x_0 = v(0)$ and
    $x_{n+1} = u(x_n, n+1)$. For example, the factorial function is
    then given by the choice $X = \Nat$;
    $u(n,m) = n \o m$ and $v(0) = 1$.

  \item\label{E:prec:2} For the set functor $F$ given by
    $FX = X \times X + 1$, coalgebras $\gamma\colon C \times C + 1$
    are deterministic systems with a state set $C$, a binary input and with halting
    states (expressed by $\gamma^{-1}(1)$).

    The coalgebra $\Nat$ of natural numbers with halting states $0$
    and $1$ and input structure $\gamma\colon n \mapsto (n-1, n-2)$ for
    $n \geq 2$ is parametrically recursive (see \autoref{E:wf-prec}).

    For example, to define the Fibonacci sequence starting with $a_0,
    a_1 \in \Nat$, consider the morphism
    $e\colon  F\Nat \times \Nat \cong \Nat^3 + \Nat \to \Nat$ with
    \[
      e(i,j,k) = i+j
      \quad\text{and}\quad
      e(n) = \begin{cases}
        a_0 & \text{if $n=0$}, \\
        a_1 & \text{if $n=1$}, \\
        0 & \text{if $n \geq 2$}.
      \end{cases}
    \]
    We know that there is a unique sequence $\sol e\colon \Nat \to \Nat$
    such that the diagram \eqref{eq:prec} commutes, which means
    $\sol e(0) = a_0$, $\sol e(1) = a_1$ and
    $\sol e(n+2) = \sol e(n+1) + \sol e(n)$.
    
  \item\label{E:prec:3}%
    Capretta et al.~\cite{cuv09} showed how to obtain
    Quicksort using parametric recursivity. Let $A$ be any linearly
    ordered set (of data elements). Then Quicksort is usually defined
    as the recursive function $q\colon A^* \to A^*$ given by
    \[
      \qsort(\eps) = \eps \qquad\text{and}\qquad
      \qsort(aw) = \qsort(w_{\leq a}) \star (a \qsort(w_{> a})),
    \]
    where $A^*$ is the set of all lists on $A$, $\varepsilon$ is the
    empty list, $\star$ is the concatenation of lists and $w_{\leq a}$
    denotes the list of those elements of $w$ which
    are less than or equal than $a$; analogously for $w_{> a}$.

    Now consider the functor $FX = 1 + A \times X \times X$ on $\Set$,
    where $1 = \set{\bullet}$, and form the coalgebra
    $\qsplit\colon A^* \to 1 + A \times A^* \times A^*$ given by
    \begin{equation}\label{eq:split}
      \qsplit(\eps) = \bullet
      \qquad\text{and}\qquad
      \qsplit(aw) = (a, w_{\leq a}, w_{> a})
      \qquad\text{for $a \in A$ and $w\in A^*$}.
    \end{equation}
    We shall see that this coalgebra is recursive in
    \autoref{E:wf-prec}. Thus, for the $F$-algebra
    $\qmerge\colon 1 +  A \times A^* \times A^* \to A^*$ given by
    \[
      \qmerge(\bullet) = \varepsilon
      \qquad\text{and}\qquad
      \qmerge(a,w,v) = w \star (av) 
    \]
    there exists a unique function $\qsort$ on $A^*$ such that
    $\qsort = \qmerge \cdot F\qsort \cdot \qsplit$. Notice that the
    last equation reflects the idea that Quicksort is a
    divide-and-conquer algorithm. The coalgebra structure
    $\qsplit$ divides a list into two parts $w_{\leq a}$ and
    $w_{> a}$.   Then $F\qsort$ sorts these two smaller lists, and
    finally in the combine- (or conquer-) step,
    the algebra structure $\qmerge$ merges the two sorted parts to
    obtain the desired whole sorted list.

    Similarly, functions defined by parametric recursivity
    (cf.~Diagram~\eqref{eq:prec}), can be understood as
    divide-and-conquer algorithms, where the combine-step is allowed
    to access the original parameter additionally.  For instance, in
    the current example the divide-step
    $\langle\qsplit, id_{A^*}\rangle$ produces the pair consisting of
    $(a, w_{\leq a}, w_{> a})$ and the original parameter $aw$,
    and the combine-step, which is given by an algebra
    $F X \times A^*\to X$ will, by the commutativity of
    \eqref{eq:prec}, get $aw$ as its right-hand input.
  \end{enumerate}
\end{examples}

Jeannin et al.~\cite[Sec.~4]{JeanninEA17} provide a number of
recursive functions arising in programming that are determined by
recursivity of a coalgebra, e.g.~the gcd of integers, the Ackermann
function, and the Towers of Hanoi.
\smnote{They say that alternating Turing machines yield a
  ``non-wellfounded example''; so I guess this means non-example?}

\section{The Next Time Operator and Well-Founded Coalgebras}\label{S:wfd}

As we have mentioned in the Introduction, the main issue of this paper is
the relationship between two concepts pertaining to coalgebras:
recursiveness and well-foundedness.
The concept of well-foundedness is well-known for directed graphs: it
means that the graph has no infinite directed paths. Similarly for
relations:\smnote{Graphs and relations may be the same (but in fact,
  they are not); however, a graph is well-founded iff the opposite(!)
  of its edge relation is well-founded, see \autoref{E:reco}.}  for
example, the elementhood relation $\in$ of set theory is well-founded;
this is precisely the Foundation Axiom.

Taylor~\cite[Def.~6.2.3]{taylor2}\smnote{Osius did
  \emph{not} have any general category theoretic formulation of well-foundedness!}
gave a more general category theoretic formulation of well-foundedness. We observe
here that his definition can be presented in a compact way, by using
an operator that generalizes the way one thinks of the
semantics of the `next time' operator of temporal logics
for non-deterministic (or even probabilistic) automata and transitions
systems. It is also strongly related to the
algebraic semantics of modal logic, where one passes from a graph $G$
to a function on $\pow G$. Jacobs~\cite{Jacobs02}
defined and studied the `next time' operator on coalgebras for Kripke
polynomial set functors, which can be generalized to arbitrary functors
as follows.

Recall that $\Sub(A)$ denotes the complete lattice of subobjects of $A$.
\begin{notheorembrackets} 
\begin{defn}[{\cite[Def.~8.9]{amm18}}]\label{D-tilde}
  Every coalgebra $\alpha\colon A\to FA$ induces an endofunction on $\Sub(A)$,
  called the \emph{next time operator}
  \[
    \nexttime\colon  \Sub(A)\to \Sub(A),
    \qquad
    \nexttime (s) = \pback\alpha(Fs)
    \quad\text{for $s\in \Sub(A)$}.
  \]
  In more detail: we define $\nexttime s$ and $\alpha(s)$ by the
  following pullback: 
  \begin{equation}\label{rdj}
    \begin{tikzcd}
      \nexttime S
      \arrow{r}{\alpha(s)}
      \arrow[>->]{d}[swap]{\nexttime s}
      \pullbackangle{-45}
      &
      FS
      \arrow[>->]{d}{Fs}
      \\
      A\arrow{r}{\alpha}
      &
      FA
    \end{tikzcd}
  \end{equation}
  In words, $\nexttime$ assigns to each subobject $s\colon S\monoto A$
  the inverse image of $F s$ under $\alpha$.  Since $Fs$ is a
  monomorphism, $\nexttime s$ is a monomorphism and $\alpha(s)$ is
  (for every representation $\nexttime s$ of that subobject of $A$)
  uniquely determined.
\end{defn}
\end{notheorembrackets}
\begin{expl}\label{E-ptilde}
  \begin{enumerate}
  \item Let $A$ be a graph, considered as a coalgebra for
    $\pow\colon\Set\to\Set$.  If $S \subseteq A$ is a set of vertices,
    then $\nexttime S$ is the set of vertices all of whose successors
    belong to $S$.
    
  \item\label{E-ptilde:2} For the set functor $FX = \pow(\Sigma \times X)$ expressing
    labelled transition systems the operator $\nexttime$ for a coalgebra
    $\alpha\colon A\to \pow(\Sigma \times A)$ is the semantic counterpart of the
    next time operator of classical linear temporal logic, see
    e.g.~Manna and Pn\"ueli~\cite{MP92}. In fact, for a subset $S \subto A$
    we have that $\nexttime S$ consists of those states whose next
    states lie in $S$, in symbols: 
    \[
      \nexttime S
      =
      \big\{x \in A \mid \text{$(s,y)\in \alpha(x)$ implies $y\in S$, for all $s \in \Sigma$}\big\}.
    \]
  \end{enumerate}
\end{expl}
The next time operator allows a compact definition of well-foundedness
as characterized by Taylor~\cite[Exercise~VI.17]{taylor2}
(see also~\cite[Corollary~2.19]{amms}): 
\begin{defn}\label{D:well-founded}
  A coalgebra is \emph{well-founded} if $\id_A$ is the only fixed
  point of its next time operator.
\end{defn}
\begin{rem}\label{R:fixed}
  \begin{enumerate}
  \item\label{R:fixed:1} Let us call a subcoalgebra
    $m\colon (B, \beta) \monoto (A,\alpha)$ \emph{cartesian}
    provided that the square below is a pullback. 
    \begin{equation}\label{diag:cart}
      \begin{tikzcd}
        B
        \arrow{r}{\beta} \arrow[>->]{d}[swap]{m}
        \pullbackangle{-45}
        &
        FB
        \arrow[>->]{d}{Fm}
        \\
        A \arrow{r}{\alpha}
        &
        FA
      \end{tikzcd}      
    \end{equation}
    Then $(A,\alpha)$ is well-founded iff it has no proper
    cartesian subcoalgebra.  That is, if
    $m\colon (B, \beta) \monoto (A,\alpha)$ is a cartesian subcoalgebra,
    then $m$ is an isomorphism. Indeed, the fixed points of next time
    are precisely the cartesian subcoalgebras (see \autoref{L:next}
    for a more refined statement). 

  \item\label{R:fixed:2} A coalgebra is well-founded iff
    $\nexttime$ has a unique pre-fixed point $\nexttime m \leq
    m$. Indeed, since $\Sub(A)$ is a complete lattice, the least fixed point of a
    monotone map is its least pre-fixed point. Taylor's
    definition~\cite[Def.~6.3.2]{taylor2} uses that property: he
    calls a coalgebra well-founded iff $\nexttime$ has no proper
    subobject as a pre-fixed point.
  \end{enumerate}
\end{rem}
\begin{examples}\label{E-well-founded}
  \begin{enumerate}
  \item\label{E-well-founded:0} A coalgebra for $\Pow$ regarded as a
    graph (see \autoref{E:graph}) is well-founded iff it has no
    infinite directed path, see~\cite[Example~6.3.3]{taylor2}.
    
  \item\label{E-well-founded:1} If $\mu F$ exists, then as a coalgebra
    it is well-founded. Indeed, in every pullback~\eqref{diag:cart},
    since $\iota^{-1}$ (as $\alpha$) is invertible, so is $\beta$. The
    unique algebra homomorphism from $\mu F$ to the algebra
    $\beta^{-1}\colon FB \to B$ is clearly inverse to $m$.
    
  \item\label{E-well-founded:2} If a set functor $F$ fulfils
    $F\emptyset = \emptyset$, then the only well-founded coalgebra is
    the empty one. Indeed, this follows from the fact that the empty
    coalgebra is a fixed point of $\nexttime$. For example, a
    deterministic automaton over the input alphabet $\Sigma$, as a
    coalgebra for $FX = \{0, 1\} \times X^\Sigma$, is well-founded iff
    it is empty.
    
  \item\label{E-well-founded:3} A non-deterministic automaton may be
    considered as a coalgebra for the set functor
    $FX = \{0,1\} \times (\P X)^\Sigma$. It is well-founded iff its
    state transition graph is well-founded (i.e.~has no infinite
    path). This follows from \autoref{C:cangr} below.
    
  \item\label{E-well-founded:4} Well-founded linear weighted
    automata. A linear weighted automaton, i.e.~a coalgebra
    $(A, \alpha)$ for $FX = K \times X^\Sigma$ on $\KVec$, is
    well-founded iff every path in its state transition graph
    eventually leads to $0$. This means that every path starting in a
    state $s \in A$ leads to the state $0$ after finitely many steps
    (where it stays).  In fact, denote by $a^*\colon A^* \monoto A$
    the subset of all states with that property. Clearly, $A^*$ is a
    subspace of $A$. Furthermore, $\nexttime$ preserves joins of
    $\omega$-chains in $\Sub(A)$
    (see~\autoref{R:omega}\ref{R:omega:2}). Hence, it follows from
    Kleene's fixed point theorem that the least fixed point of
    $\nexttime$ is $\bigvee_{n \in \Nat} \nexttime^n(\bot_A)$. We also
    know that $\bot_A$ is the $0$-subspace, and for every subspace
    $s\colon S \monoto A$, $\nexttime s$ is the space of all nodes
    whose successors are in $S$. Therefore $\nexttime^n(\bot_A)$
    consists of precisely those states from which every path reaches
    $0$ in at most $n$ steps. Thus
    $A^* = \bigvee_{n \in \Nat} \nexttime^n(\bot_A)$. It follows that
    $(A, \alpha)$ is well-founded iff $A = A^*$.
  \end{enumerate}
\end{examples}

We next show that to every coalgebra for a set functor $F$ one may
associate a graph, in a canonical way.  Moreover, if $F$ preserves
intersections, then a coalgebra is well-founded if and only if so is
its canonical graph.

\begin{notation}
  Given a set functor $F$, we define for every set $X$
  the map $\tau_X\colon FX \to \pow X$ assigning to every element
  $x \in FX$ the intersection of all subsets $m\colon M \subto X$ such
  that $x$ lies in the image of $Fm$:
  \begin{equation}\label{eq:tau}
    \tau_X(x) = \bigcap \{m \mid \text{$m\colon M \subto X$ satisfies $x \in Fm[FM]$}\}.
  \end{equation}
\end{notation}
\begin{defn}\label{D:can}
  Let $F$ be a set functor. For every coalgebra $\alpha\colon A\to FA$
  its \emph{canonical graph} is the following
  coalgebra for $\pow$:
  \[
    A \xrightarrow{\alpha} FA \xrightarrow{\tau_A} \pow A.
  \]
\end{defn}
\begin{examples}\label{E:tau}
  \begin{enumerate}
  \item Given a graph as a coalgebra $\alpha\colon A \to \pow A$, the condition
    $\alpha(x) \in \pow m [\pow M]$ states precisely that all
    successors of $x$ lie in the set $M$. The least such set is
    $\alpha(x)$. Therefore, the canonical graph of $(A,\alpha)$ is
    itself (see~\cite[Example~6.3.3]{taylor2}).

  \item For the type functor of $FX = \set{0,1} \times X^\Sigma$ of
    deterministic automata, we have
    \[
      \tau_X(i,t) = \set{t(s) : s \in \Sigma}
      \qquad \text{for $i = 0,1$ and $t\colon \Sigma \to X$}.      
    \]
    Thus, the canonical graph of a deterministic automaton $A$ is
    precisely its state transition graph (forgetting the labels
    of transitions and the finality of states), i.e.~we have an
    edge $(a,a')$ iff  $a' = \delta(a,s)$ for some $s \in \Sigma$,
    where $\delta$ is the nextstate function of $A$.

    Similarly, for the type functor
    $FX = \{0,1\} \times (\P X)^\Sigma$ of non-deterministic automata
    we have 
    \[
      \tau_X (i,g) = \bigcup\limits_{s \in \Sigma} t(s)
      \qquad \text{for $i = 0,1$ and $t\colon \Sigma \to \P X$}.
    \]

  \item For the functor $FX = \pow(\Sigma \times X)$ whose coalgebras are labeled
    transition systems we have
    \[
      \tau_X = (\pow(\Sigma \times X) \xrightarrow{\pow \pi_X} \pow
      X),
    \]
    where $\pi_X\colon \Sigma \times X \to X$ is the projection. Again,
    the canonical graph of a labelled transition system is its state
    transition graph. Thus $(a,a')$ is an edge iff some action leads
    from state $a$ to $a'$.
  \end{enumerate}
\end{examples}
Recall that a functor \emph{preserves intersections} if it preserves
(wide) pullbacks of families of monomorphisms.
Gumm~\cite[Theorem~7.3]{Gumm2005} observed that for a set functor
preserving intersections,  the maps $\tau_X\colon FX \to \pow X$
in~\eqref{eq:tau} form a ``subnatural'' transformation from $F$ to the
power-set functor $\pow$. Subnaturality means that (although these
maps do not form a natural transformation in general) for every
monomorphism $i\colon X \to Y$ we have a commutative square:
\begin{equation}\label{diag:subcar}
  \begin{tikzcd}
    FX
    \pullbackangle{-45}
    \arrow{r}{\tau_X}
    \arrow[>->]{d}[swap]{Fi}
    &
    \pow X
    \arrow[>->]{d}{\pow i}
    \\
    FY
    \arrow{r}{\tau_Y}
    &
    \pow Y
  \end{tikzcd}
\end{equation}
For many set functors this is even a pullback square:
\begin{notheorembrackets}%
\begin{theorem}[{\cite[Thm.~7.4]{Gumm2005} and \cite[Prop.~7.5]{taylor3}}]
  A set functor $F$ preserves intersections iff the 
  squares in~\eqref{diag:subcar} above are pullbacks. 
\end{theorem}
\end{notheorembrackets}
\begin{notheorembrackets}%
\begin{theorem}[{\cite[Thm.~8.1]{Gumm2005} and \cite[Prop.~7.5]{taylor3}}]\label{LGumm}
  Let $F$ be a set functor which preserves inverse images and intersections. 
  Then $\tau\colon F\to \pow$ is a natural transformation.
\end{theorem}
\end{notheorembrackets}
\begin{expl} 
  To see that $\tau$ is not a natural transformation in general, we
  use the set functor $R$ from Example~\ref{E:setfunctors}, part~\ref{E:setfunctors:2}.
   Let
  $X = \set{0,1}$, $Y = \set{0}$, and $f\colon X\to Y$ the unique function.
  Then $(0,1)\in FX$, and $\tau_X(0,1) = X$.  Further,
  $\pow f(X) = Y$.  But $Rf(0,1) = d$, and $\tau_Y(d) = \emptyset$.
\end{expl}
\begin{lemma}\label{R:nexttime:2}
  For every set functor $F$ preserving intersections, the next time
  operator of a coalgebra $(A,\alpha)$ coincides with that of its
  canonical graph.
\end{lemma}
\begin{proof}
  In the diagram below the outside is a pullback if and only if so is
  the left-hand square:
  \[
    \begin{tikzcd}
      \nexttime A'
      \arrow[dashed]{r}{\alpha(m)}
      \arrow[shiftarr={yshift=20pt}]{rr}{\tau_{A'} \cdot \alpha(m)}
      \arrow[>->]{d}[swap]{\nexttime m}
      \pullbackangle{-45}
      &
      FA'
      \arrow{d}[swap]{Fm}
      \arrow{r}{\tau_{A'}}
      \pullbackangle{-45}
      &
      \Pow A'
      \arrow{d}{\Pow m}
      \\
      A
      \arrow{r}{\alpha}
      &
      FA
      \arrow{r}{\tau_A}
      &
      \Pow A
    \end{tikzcd}
    \tag*{\qed}
  \]
\end{proof}

Taylor~\cite[Rem.~6.3.4]{taylor2} proved the following result for functors
preserving intersections and inverse images; the latter assumption is
not needed.
\begin{notheorembrackets}%
\begin{corollary}[\cite{taylor2}]\label{C:cangr}
  A coalgebra for a set functor preserving intersections is
  well-founded iff its canonical graph is well-founded.
\end{corollary}
\end{notheorembrackets}
\begin{examples}\label{E:cangr}
  \begin{enumerate}
  \item A coalgebra for the identity functor $FX = X$ on $\Set$ is a
    set $A$ equipped with a function $\alpha\colon A \to A$. The
    canonical graph of $(A,\alpha)$ is the graph of the function
    $\alpha$, i.e.~the graph with edges $(a,\alpha(a))$ for all
    $a \in A$. Hence, $(A,\alpha)$ is well-founded iff it is empty
    (see \autoref{E-well-founded}\ref{E-well-founded:2}).

  \item For $FX = X + 1$ coalgebras are sets $A$ equipped
    with a \emph{partial} function $\alpha\colon A \to A$, and the
    canonical graph is the graph of $\alpha$. This functor has many
    nonempty well-founded coalgebras. For example, the initial
    $F$-algebra, considered as the coalgebra on $\Nat$ with the
    structure given by the partial function $n \mapsto n-1$, for $n>0$
    (cf.~\autoref{E:prec}\ref{E:prec:1}), is well-founded since its
    canonical graph is so.

  \item For a (deterministic or non-deterministic) automaton, the
    canonical graph has an edge from $s$ to $t$ iff there is a
    transition from $s$ to $t$ for some input letter. Thus, we obtain
    the characterization of well-foundedness as stated in
    \autoref{E-well-founded}\ref{E-well-founded:2} and~\ref{E-well-founded:3}.

    
  \item\label{E:cangr:3} Consider the functor $FX = X \times X + 1$
    and a coalgebra $\alpha\colon A \to A \times A  + 1$.  The edges in
    its canonical graph are all of the pairs $(a,a_1)$ and $(a,a_2)$
    such that $a \in A$ and $\alpha(a) = (a_1,a_2)$. For example, the
    coalgebra $(\Nat,\gamma)$ from \autoref{E:prec}\ref{E:prec:2} has
    the canonical graph with edge set
    $\set{(n,n-1),(n,n-2) : n \geq 2}$, which is clearly well-founded,
    and therefore so is the coalgebra.

    Similarly, for the functor $FX = 1 + A \times X \times X$, the
    coalgebra $(A^*,\qsplit)$ in \autoref{E:prec}\ref{E:prec:3} is easily
    seen to be well-founded via its canonical graph. Indeed, this
    graph has for every list $w$ one outgoing edge to the list
    $w_{\leq a}$ and one to $w_{> a}$ for every $a \in A$. Hence, this
    is a well-founded graph.

  \item More generally, for a polynomial functor $H_\Sigma\colon \Set
    \to \Set$ associated to a finitary signature $\Sigma$, a
    coalgebra $\alpha\colon A \to \coprod_{n \in \Nat} \Sigma_n \times
    A^n$ has the canonical graph where every
    vertex $a \in A$ has an outgoing edge $(a,a')$ for every $a' \in A$
    occurring in the tuple $\alpha(a) \in \Sigma_n \times A^n$ for some
    $n <\omega$.

    Thus, the coalgebra $(A,\alpha)$ is well-founded iff for every
    $a \in A$ its tree-unfolding, i.e., its image under the unique
    homomorphism $h\colon A \to \nu F$, is a finite $\Sigma$-tree. In
    particular, if the signature $\Sigma$ does not contain any
    constant symbols, then the only well-founded $H_\Sigma$-coalgebra
    is $A = \emptyset$.
  \end{enumerate}
\end{examples}

For further use we now compare well-founded and recursive coalgebras
for a given set functor $F$ with those of its Trnkov\'a hull
$\bar F$ (see~\autoref{P:Tr}). Since empty coalgebras are (trivially)
well-founded and recursive, we can restrict ourselves to the
nonempty ones. Observe that $\Coalg F$ and $\Coalg \bar F$ have the
same nonempty objects, and these categories are isomorphic.
\begin{lemma}\label{L:Trn}
  Let $(A,\alpha)$ be a nonempty coalgebra for a set functor $F$. If
  it is well-founded or (parametrically) recursive, then it also has
  those properties as a coalgebra for the Trnkov\'a hull $\bar F$.
\end{lemma}
\begin{proof}
  \begin{enumerate}
  \item Let $(A,\alpha)$ be well-founded for $F$. Nonempty
    subcoalgebras of $(A,\alpha)$ for $F$ and for $\bar F$
    coincide. Thus, we only need to show that the left-hand square
    below, where $r_X\colon \emptyset \to X$ denotes the empty map, is
    not a pullback:
    \[
      \begin{tikzcd}
        \emptyset
        \arrow{r}{r_{\bar F \emptyset}}
        \arrow{d}[swap]{r_A}
        &
        \bar F \emptyset
        \arrow{d}{\bar F r_A}
        \\
        A \arrow{r}{\alpha}
        &
        \bar FA = FA
      \end{tikzcd}
      \qquad
      \begin{tikzcd}
        \emptyset
        \arrow{r}{r_{F \emptyset}}
        \arrow{d}[swap]{r_A}
        &
        F \emptyset
        \arrow{d}{F r_A}
        \\
        A \arrow{r}{\alpha}
        &
        FA
      \end{tikzcd}
    \]
    Since $(A,\alpha)$ is well-founded, we know that the right-hand
    square is not a pullback. Thus, there exist $a \in A$ and
    $x \in F\emptyset$ with $\alpha(a) = Fr_A(x)$. For the functor
    $C_{01}$ of \autoref{R:inter},  define a natural transformation
    $\tau\colon C_{01} \to F$ by $\tau_X = Fr_X(x) \in FX$. Then
    $\alpha(a) = \tau_A$, and we know that $\tau$ lies in
    $\bar F \emptyset$ so that $\tau_A = \bar F
    r_A(\tau)$. Consequently, we have $\alpha(a) = \bar Fr_A(\tau)$,
    which proves that the left-hand square above is not a pullback.
    
  \item Let $(A,\alpha)$ be a nonempty recursive coalgebra for
    $F$. Given an algebra $e\colon \bar F X \to X$ we know that
    $X \neq \emptyset$, for otherwise the existence of a unique
    coalgebra-to-algebra morphism $A \to X$ would force $A$ to be
    empty. But then the unique coalgebra-to-algebra morphism from
    $(A,\alpha)$ to $(X,e)$ w.r.t.~$F$ is also one for $\bar F$.\qed
  \end{enumerate}
\end{proof}

We now collect a few properties of the next time operator
we will need in the following. 
\begin{lemma}\label{R:nexttime}
  The next time operator is monotone: if $m \leq n$,
  then $\nexttime m\leq \nexttime n$.
\end{lemma}
\begin{proof}
  Suppose that $m\colon A' \monoto A$ and $n\colon A'' \monoto A$ are
  subobjects such that $m \leq n$, i.e.~$n \cdot x = m$ for some
  $x\colon A' \monoto A''$. Then we obtain the dashed arrow in the
  diagram below using that its lower square is a pullback:
  \[
    \begin{tikzcd}
      \nexttime A'
      \arrow[dashed]{d}
      \arrow{r}{\alpha(m)}
      \arrow[>->,shiftarr={xshift=-10mm}]{dd}[swap]{\nexttime m}
      &
      FA'
      \arrow{d}{Fx}
      \arrow[>->,shiftarr={xshift=10mm}]{dd}{Fm}
      \\
      \nexttime A''
      \arrow{r}{\alpha(n)}
      \arrow[>->]{d}[swap]{\nexttime n}
      \pullbackangle{-45}
      &
      FA''
      \arrow[>->]{d}{Fn}
      \\
      A
      \arrow{r}{\alpha}
      &
      FA
    \end{tikzcd}
  \]
  This shows that  $\nexttime m\leq \nexttime n$. \qed
\end{proof}
The following lemma will be useful when we establish the universal
property of the well-founded part of a coalgebra in the next section. 
\begin{lemma}\label{L:amb}
  For every coalgebra homomorphism $f\colon (B,\beta) \to (A,\alpha)$
  we have
  \[
    \nexttime_\beta\o \pback f \leq \pback f \o \nexttime_\alpha,
  \]
  where $\nexttime_\alpha$ and $\nexttime_\beta$ denote the next time
  operators of the coalgebras $(A,\alpha)$ and $(B,\beta)$,
  respectively, and $\leq$ is the pointwise order.
\end{lemma}
\begin{proof}
  Let $s\colon S \monoto A$ be a subobject. We see that
  $\pback f (\nexttime_\alpha s)$ is obtained by pasting two pullback
  squares as shown below:
  \begin{equation}\label{wfpbs}
    \begin{tikzcd}
      T
      \arrow{r}{t}
      \arrow[>->]{d}[swap]{\pback f(\nexttime_\alpha s)}
      \pullbackangle{-45}
      &
      \nexttime_\alpha S
      \arrow{r}{\alpha(s)}
      \arrow[>->]{d}[swap]{\nexttime_\alpha s}
      \pullbackangle{-45}
      &
      FS
      \arrow[>->]{d}{Fs}
      \\
      B \arrow{r}{f}
      &
      A
      \arrow{r}{\alpha}
      &
      FA
    \end{tikzcd}
  \end{equation}
  In order to show that $\nexttime_\beta(\pback f(s)) \leq  \pback f
  (\nexttime_\alpha s)$, we consider the following diagram:
  \begin{equation}\label{diag:amb}
    \begin{tikzcd}[column sep = 15mm]
      \nexttime_\beta U
      \arrow{r}{\beta(\pback f (s))}
      \arrow[>->]{dd}[swap]{\nexttime_\beta (\pback f (s))}
      \pullbackangle{-45}
      &
      FU
      \arrow[>->]{d}[swap]{F(\pback f (s))}
      \arrow{r}{Fu}
      &
      FS
      \arrow[>->]{dd}{Fs}
      \\
      &
      FB
      \arrow{rd}{Ff}
      \\
      B 
      \arrow{r}{f}
      \arrow{ru}{\beta}
      &
      A
      \arrow{r}{\alpha}
      &
      FA
    \end{tikzcd}
  \end{equation}
  The upper left-hand part is the pullback square defining
  $\nexttime_\beta(\pback f(s))$, and the upper right-hand one is that
  defining $\pback f (s)$, with $F$ applied.  On the bottom, we use
  that $f$ is a coalgebra homomorphism. Thus, the outside of the diagram
  commutes. Since the outside of the diagram in \eqref{wfpbs} is a
  pullback, we have some $g\colon \nexttime_\beta U \to T$ such 
  that $\nexttime_\beta(\pback f(s)) = \pback f(\nexttime_\alpha(s))
  \cdot g$, which proves the desired inequality. \qed
\end{proof}
\begin{corollary}\label{C:pback}
  For every coalgebra homomorphism $f\colon (B, \beta) \to
  (A,\alpha)$ we have $\nexttime_\beta \cdot \pback f = \pback f \cdot
  \nexttime_\alpha$ provided that either
  \begin{enumerate}
  \item\label{C:pback:1} $f$ is a monomorphism in
    $\A$ and $F$ preserves finite intersections, or
  \item\label{C:pback:2} $F$ preserves inverse images.
  \end{enumerate}
\end{corollary}
\begin{proof}
  Indeed, under either of the above conditions, the upper right-hand part in
  Diagram~\eqref{diag:amb} is a pullback. Thus, pasting this part with
  the pullback in the upper left of~\eqref{diag:amb} and using that the
  lower part commutes, we see that $\nexttime_\beta(\pback f(s))$ is
  obtained by pulling back $Fs$ along $f \cdot \alpha$. This implies
  the desired equality since this is how $\pback f (\nexttime_\alpha
  s)$ is obtained (see~\eqref{wfpbs}).\qed
\end{proof}

\begin{lemma}\label{L:next}
  Let $\alpha\colon A\to FA$ be a coalgebra and $m\colon B \monoto A$ be a
  monomorphism.
  \begin{enumerate}
  \item\label{L:next:1}  There is a coalgebra structure 
    $\beta\colon B\to FB$ for  which $m$ gives a subcoalgebra of $(A,\alpha)$ iff
    $m \leq \nexttime m$.
  \item\label{L:next:2} There is a coalgebra structure $\beta\colon B\to FB$ for
    which $m$ gives a cartesian subcoalgebra of $(A,\alpha)$ iff
    $m = \nexttime m$.
  \end{enumerate}  
\end{lemma}
\begin{proof}
  We prove the left-to-right directions of both assertions first, and
  then the right-to-left ones.

  Suppose first that there exists $\beta\colon B \to FB$ such that
  $m\colon (B,\beta)\monoto (A,\alpha)$ is a coalgebra morphism.  Then
  the fact that $\nexttime B$ is given by a pullback yields a morphism
  $x\colon B \monoto \nexttime B$ such that, inter alia,
  $\nexttime m\o x = m$.  It follows that $m \leq \nexttime m$.  If
  $(B,\beta)$ is a cartesian subcoalgebra, then we have a pullback
  square
  \[
    \begin{tikzcd}
      B
      \arrow{r}{\beta}
      \arrow[>->]{d}[swap]{m}
      \pullbackangle{-45}
      &
      FB
      \arrow{d}{Fm}
      \\
      A \arrow{r}{\alpha}
      &
      FA
    \end{tikzcd}
  \]
  So clearly  $m = \nexttime m$ in $\Sub(A)$.

  Conversely, suppose that $m \leq \nexttime m$ via $x\colon B
  \monoto \nexttime B$.
  Then $\alpha(m)\o x\colon B\to FB$ is a coalgebra,
  and $m\colon B \monoto A$ is a homomorphism:
  \[
    \begin{tikzcd}[row sep = 25pt]
      B
      \arrow{r}{x}\arrow[>->]{rd}[swap]{m}
      &
      \nexttime B
      \arrow[>->]{d}{\nexttime m}
      \arrow{r}{\alpha(m)}
      \pullbackangle{-45}
      &
      FB
      \arrow{d}{Fm}
      \\
      &
      A\arrow{r}{\alpha}
      &
      FA
    \end{tikzcd}
  \]
  If in addition $m = \nexttime m$, i.e.~$x$ is an isomorphism, we see
  that $m$ is a cartesian subcoalgebra. \qed
\end{proof}

We close this section with a characterization result: $F$ preserves
intersections if and only if the following ``generalized next time''
operators are right adjoints. Given a morphism $f\colon A \to FB$, we
have the operator $\nexttime_f\colon \Sub(B) \to \Sub(A)$ that maps
$m\colon B' \monoto B$ to the pullback of $Fm$ along $f$:
\[
  \begin{tikzcd}    
    \nexttime_f A' \arrow[>->]{d}[swap]{\nexttime_f m} \arrow{r}{f(m)}
    \pullbackangle{-45}
    &
    FB' \arrow[>->]{d}{Fm} \\
    A \arrow{r}{f} & FB
  \end{tikzcd}
\]
\begin{notheorembrackets}%
\begin{proposition}[\cite{WissmannEA19}]\label{P:adjnext}
  The functor $F$ preserves intersections if and only if every
  generalized next time operator $\nexttime_f$ is a right adjoint.
\end{proposition}
\end{notheorembrackets}
\begin{proof}
  For the ``if''-direction, choose
  $f = \id_{FY}$. Then $\nexttime_{\id_{FY}}\colon m \mapsto Fm$ is a
  right adjoint and so preserves all meets, i.e.~$F$ preserves
  intersections.

  The converse follows from the easily established fact that
  intersections are stable under inverse image, i.e.~for every morphism
  $f\colon X \to Y$ and every family $m_i\colon S_i \monoto Y$ of
  subobjects, the intersection $m\colon P \monoto X$ of the inverse images
  of the $m_i$ under $f$ yields a pullback
  \[
    \begin{tikzcd}
      P
      \arrow[>->,swap]{d}{m} \arrow{r} \pullbackangle{-45} & \bigcap S_i
      \arrow[>->]{d}{\bigcap m_i}\\
      X \arrow{r}{f} & Y
    \end{tikzcd}
  \]
  Hence, if $F$ preserves intersections, then so does every operator
  $\nexttime_f$. Equivalently, $\nexttime_f$ is a right adjoint. \qed
\end{proof}

\section{The Well-Founded Part of a Coalgebra}
\label{section-well-founded-part}

We introduced well-founded coalgebras in
\autoref{S:wfd}. We now discuss the
\emph{well-founded part} of a coalgebra, i.e.~its largest well-founded
subcoalgebra. We prove that this is the least fixed point of the
next time operator.  Then we prove that the
well-founded part is the coreflection of a coalgebra in the category
of well-founded coalgebras.
\begin{notheorembrackets}
\begin{defn}[\cite{amm18}]
  The \emph{well-founded part} of a coalgebra is its largest
  well-founded subcoalgebra.
\end{defn}
\end{notheorembrackets}

The well-founded part of a coalgebra always exists and is the
coreflection in the category of well-founded
coalgebras~\cite[Prop.~2.27]{amms}. We provide a new, shorter proof of
this fact. The well-founded part is obtained by the following:
\begin{notheorembrackets}
\begin{construction}[{\cite[Not.~2.22]{amms}}]
  Let $\alpha\colon A \to FA$ be a coalgebra. We know that $\Sub(A)$
  is a complete lattice and that the next time operator $\nexttime$ is monotone
  (see \autoref{R:nexttime}). Hence, by the Knaster-Tarski fixed point
  theorem, $\nexttime$ has a
  least fixed point, which we denote by
  \[
    a^*\colon A^* \monoto A.
  \]
  Moreover, by \autoref{L:next}\ref{L:next:2}, we know that there is a coalgebra
  structure $\alpha^*\colon A^* \to FA^*$ so that
  $a^*\colon (A^*, \alpha^*) \monoto (A, \alpha)$ is the smallest cartesian
  subcoalgebra of $(A, \alpha)$.
\end{construction}
\end{notheorembrackets}
\begin{proposition}\label{P:wfdpart}
  For every coalgebra $(A,\alpha)$, the coalgebra $(A^*, \alpha^*)$ is
  well-founded. 
\end{proposition}
\begin{proof}
  Let $m\colon (B, \beta) \monoto (A^*, \alpha^*)$ be a cartesian
  subcoalgebra. By \autoref{L:next}, $a^*\o m\colon B\to A$ is a fixed
  point of $\nexttime$. Since $a^*$ is the least fixed point, we have
  $a^* \leq a^*\o m$, i.e.~$a^* = a^* \cdot m \cdot x$ for some
  $x\colon A^* \monoto B$. Since $a^*$ is monic, we thus have
  $m \cdot x = \id_{A^*}$.  So $m$ is a monomorphism and a split epimorphism, whence
  an isomorphism.\qed
\end{proof}
\begin{expl}
  Consider the coalgebra $G$ for $\pow$ depicted as
  the following graph:
  \[
    \xymatrix{
      a\ar[r] & b  &  c \ar@/_.5pc/[r] & \ar@/_.5pc/[l] d
    }
  \]
  We list all subcoalgebras below (the structures are the obvious ones
  given by the picture of $G$).  Those are $\emptyset$,
  $\set{b}$, $\set{a,b}$, $\set{c,d}$, $\set{b,c,d}$, and
  $\set{a,b,c,d}$.  Of these, the cartesian subcoalgebras of $G$ are
  $\set{a,b}$, and $\set{a,b,c,d}$. The well-founded part of $G$ is
  the least cartesian subcoalgebra, namely $\set{a,b}$.
\end{expl}
We know from \autoref{P:wfdpart} that for every coalgebra
$(A, \alpha)$ its subcoalgebra represented by
$a^*\colon A^* \monoto A$ is well-founded. We now
prove that, categorically, this subcoalgebra is characterized uniquely
up to isomorphism by the following universal property: every
homomorphism from a well-founded coalgebra into $(A,\alpha)$ factorizes
uniquely through $a^*$. In particular, this implies that
$a^*\colon A^* \to A$ is the largest well-founded subcoalgebra of $A$,
viz.~the well-founded part of $A$.
\begin{proposition}\label{P:wfdpart2}
  The full subcategory of $\Coalg F$ given by well-founded coalgebras
  is coreflective. In fact, the well-founded coreflection of a
  coalgebra is its well-founded part $a^*\colon(A^*,\alpha^*)\monoto (A,\alpha)$.
\end{proposition}
\begin{proof}
  We are to prove that for every coalgebra homomorphism
  $f\colon (B,\beta) \to (A,\alpha)$, where $(B,\beta)$ is
  well-founded, there exists a coalgebra homomorphism
  $f^\sharp\colon (B,\beta) \to (A^*,\alpha^*)$ such that
  $a^* \cdot f^\sharp = f$.  It is unique since
  $a^*\colon A^* \monoto A$ is a monomorphism. It then follows that
  $a^*\colon (A^*,\alpha^*) \monoto (A,\alpha)$ is the largest
  well-founded subcoalgebra.

  For the existence of $f^\sharp$, we first observe that $\pback
  f(a^*)$ is a pre-fixed point of $\nexttime_\beta$: indeed, using
  \autoref{L:amb} we have
  \[
    \nexttime_\beta(\pback f(a^*)) \leq \pback
    f(\nexttime_\alpha(a^*)) = \pback f(a^*).
  \]
  By  \autoref{R:fixed}\ref{R:fixed:2}, we therefore have 
  $id_B\ = b^* \leq \pback f(a^*)$ in $\Sub(B)$. Using the adjunction in
  \autoref{P:subadjs}, we have $\pbackleft f(\id_B) \leq a^*$ in
  $\Sub(A)$. Now let
  \[
    f = \big( B \overset{e}{\epito} C
    \overset{m}{\monoto} A\big)
  \]
  be the factorization of $f$ as in
  \autoref{R:compwell}\ref{R:compwell:3}. 
  This
  implies that $\pbackleft f(\id_B) = m$. Thus we obtain
  \[
    m = \pbackleft f(\id_B) \leq a^*,
  \]
  i.e.~there exists a morphism $h\colon C \monoto A^*$
  such that $a^* \cdot h = m$. Thus, $f^\sharp =
  h \cdot e\colon B \to A^*$ is a morphism satisfying 
  \[
    a^* \cdot f^\sharp = a^* \cdot h \cdot e = m \cdot e = f. 
  \]
  It follows that $f^\sharp$ is a coalgebra homomorphism from
  $(B,\beta)$ to $(A^*,\alpha^*)$ since $f$ and $a^*$ are and $F$
  preserves monomorphisms.\qed
\end{proof}
\section{Closure Properties of Well-Founded Coalgebras}
\label{S:closure}
In this section we will see that strong quotients and subcoalgebras
(see~\autoref{R:subcoalg}) of well-founded coalgebras are well-founded
again. For subcoalgebras we need to assume more about $\A$ and $F$. We
present two variants in \autoref{P:wfsub} and \autoref{T:sub}.

We mention the following corollary to \autoref{P:wfdpart2}. For
endofunctors on sets preserving inverse images this was stated by
Taylor~\cite[Exercise~VI.16]{taylor2}:
\begin{corollary}\label{C:wfcolim}
  The subcategory of $\Coalg F$ formed by all well-founded
  coalgebras is closed under strong quotients and coproducts in $\Coalg F$.
\end{corollary}
This follows from a general result on coreflective
subcategories~\cite[Thm.~16.8]{ahs}: the category $\Coalg F$ has the
factorization system of~\autoref{P:(e,m)}, and its full subcategory of
well-founded coalgebras is coreflective with monomorphic coreflections
(see \autoref{P:wfdpart2}). Consequently, it is closed under strong
quotients and colimits.
\begin{rem}
  We prove next that, for an endofunctor preserving finite
  intersections, well-founded coalgebras are closed under subcoalgebras
  provided that $\Sub(A)$ forms a frame. This assumption is not needed
  provided that monomorphisms are universally \smooth (see
  \autoref{T:sub}). Recall that $\Sub(A)$ is a frame if for every
  subobject $m\colon B \monoto A$ and every family $m_i\ (i\in I)$ of
  subobjects of $A$ we have
  \[
    m \wedge \bigvee_{i \in I} m_i = \bigvee_{i\in I} (m\wedge m_i).
  \]
  Equivalently, $\pback m\colon \Sub(A) \to \Sub(B)$ has a right
  adjoint $\pbackright m\colon \Sub(B) \to \Sub(A)$ (use the dual
  of~\autoref{R:adjoint}).
\end{rem}
\begin{examples}
  \begin{enumerate}
  \item $\Set$ has the property that all $\Sub(A)$ are frames. In
    fact, given subsets $S$ and $S_i\ (i \in I)$ of $A$ the equality
    $S \cap (\bigcup_{i \in I} S_i ) = \bigcup_{i\in I} (S \cap S_i)$
    clearly holds.

  \item This property is shared by categories such as posets and
    monotone maps, graphs and homomorphisms, unary algebras and
    homomorphisms, topological spaces and continuous maps, and
    presheaf categories $\Set^{\C^\mathsf{op}}$, with $\C$ small. This
    follows from the fact that joins and meets of subobjects of an
    object $A$ are formed on the level of subsets of the underlying
    set of $A$.

  \item For every Grothendieck topos, the posets $\Sub(A)$ are 
    frames. In fact, it is sufficient for a topos to have all
    coproducts or intersections to satisfy this requirement.

  \item The category $\KVec$ does not have the above property. For
    example, for $K = \Reals$ and two distinct lines 
    $m_i\colon \Reals\monoto \Reals^2$, the desired equation
    fails. Indeed, for every line
    $m\colon \Reals\monoto \Reals^2$ different from $m_1, m_2$ we have
    that
    \[
      m\andd (m_1\orr m_2) = m \neq 0 = (m\andd m_1)\orr (m\andd m_2).
    \]

  \item The category $\CPO$ does not have the above property: for the
    cpo $A = \Nat^\top$ of natural numbers with a top element $\top$
    (linearly ordered) the lattice $\Sub(A)$ is not a frame. Consider
    the subobjects given by inclusion maps $m_i\colon \set{0, \ldots,
      i} \subto \Nat^\top$ for $i \in \Nat$, with domains linearly
    ordered. It is easy to see that $\bigvee_{i\in \Nat} m_i =
    \id_A$. For the inclusion map $m\colon \set{\top} \subto
    \Nat^\top$ we have $m \wedge m_i = 0\ (i \in I)$, the empty
    subobject. Thus, $\bigvee_{i \in \Nat} (m \wedge m_i) = 0 \neq m =
    m \wedge \bigvee_{i \in I} m_i$.
  \end{enumerate}
\end{examples} 
\begin{proposition}\label{P:wfsub}
  Suppose that $F$ preserves finite intersections, and let 
  $(A,\alpha)$ be a well-founded coalgebra such that $\Sub(A)$ a
  frame. Then every subcoalgebra of $(A,\alpha)$ is well-founded.
\end{proposition}
\begin{proof}
  Let $m\colon (B,\beta) \monoto (A,\alpha)$ be a subcoalgebra. We
  will show that the only pre-fixed point of $\nexttime_\beta$ is
  $\id_B$ (cf.~\autoref{R:fixed}\ref{R:fixed:2}). Suppose
  $s\colon S \monoto B$ fulfils $\nexttime_\beta(s) \leq s$. Since $F$
  preserves finite intersections, we have
  \[
    \pback m \cdot \nexttime_\alpha = \nexttime_\beta \cdot \pback m
  \]
  by \autoref{C:pback}\ref{C:pback:1}. The counit of the adjunction $\pback m \dashv
  \pbackright m$ yields $\pback m (\pbackright m (s)) \leq s$, so that
  we obtain
  \[    
    \pback m(\nexttime_\alpha(\pbackright m(s))) =
    \nexttime_\beta(\pback m (\pbackright m(s))) \leq
    \nexttime_\beta(s) \leq s.
  \]
  Using again the adjunction $\pback m \dashv \pbackright m$, we have
  equivalently that
  $\nexttime_\alpha(\pbackright m (s)) \leq \pbackright m(s)$,
  i.e.~$\pbackright m(s)$ is a pre-fixed point of
  $\nexttime_\alpha$. Since $(A,\alpha)$ is well-founded,
  \autoref{C:pback}\ref{C:pback:1} implies that
  $\pbackright m (s) = \id_A$. Since $\pback m$ is also a right
  adjoint and therefore preserves the top element of $\Sub(B)$, we
  thus obtain
  \[
    \id_B = \pback m(\id_A) = \pback m(\pbackright m(s)) \leq s,
  \]
  which completes the proof.\qed
\end{proof}
\begin{rem}
  Given a set functor $F$ preserving inverse images, a much better
  result was proved by Taylor~\cite[Corollary~6.3.6]{taylor2}: for
  every coalgebra homomorphism $f\colon (B,\beta) \to (A,\alpha)$ with
  $(A,\alpha)$ well-founded so is $(B,\beta)$. In fact, our proof
  above is essentially Taylor's who (implicitly) uses
  \autoref{C:pback}\ref{C:pback:2} instead.
\end{rem}
\begin{corollary}\label{C:setsub}
  If a set functor preserves finite intersections, then subcoalgebras
  of well-founded coalgebras are well-founded.
\end{corollary}

Trnkov\'{a} proved~\cite{trnkova69} that every set functor
preserves all \emph{nonempty} finite intersections.  However, this
does not suffice for Corollary~\ref{C:setsub}: 
\begin{expl}
  A well-founded coalgebra for a set functor can have non-well-founded
  subcoalgebras. Let $F\emptyset = 1$ and $FX = 1+1$ for all nonempty
  sets $X$, and let $Ff = \inl\colon 1 \to 1+1$ be the left-hand
  injection for all maps $f\colon \emptyset \to X$ with $X$
  nonempty. The coalgebra $\inr\colon 1 \to F1$ is not well-founded
  because its empty subcoalgebra is cartesian. However, this is a
  subcoalgebra of $\id\colon 1+1 \to 1+1$ (via the embedding $\inr$),
  and the latter is well-founded.
\end{expl}

The fact that subcoalgebras of a well-founded coalgebra are
well-founded does not necessarily need the assumption that $\Sub(A)$
is a frame. Using the construction of the least fixed point $a^*$ of
$\nexttime$ provided by the (proof of the) Knaster-Tarski fixed point
theorem, it is essentially sufficient that $\pback m$ in the proof of
\autoref{P:wfsub} preserves joins of unions of chains in $\Sub(A)$. We now
discuss this in more detail.

Recall (universally) \smooth monomorphisms from \autoref{D:constr}.
\begin{notheorembrackets}
\begin{construction}[{\cite[Not.~2.22]{amms}}]\label{C:wf-part}
  Let $(A,\alpha)$ be a coalgebra. We obtain $a^*$, the least fixed
  point of $\nexttime$, as the join of the following transfinite
  chain of subobjects $a_i\colon A_i \monoto A$, $i \in \Ord$. First,
  put $a_0 = \bot_A$, the least subobject of
  $A$. Given $a_i\colon A_i \monoto A$, put
  $a_{i+1}= \nexttime a_i \colon A_{i+1} = \nexttime A_i \monoto A$.
  For every limit ordinal $j$, put $a_j = \bigvee_{i < j}
  a_i$. It follows from the proof of the Knaster-Tarski fixed point theorem
  that there exists an ordinal $i$ such that $a_i = a^*\colon A^* \monoto A$.
\end{construction}
\end{notheorembrackets}
\begin{rem}\label{R:omega}
  \begin{enumerate}
  \item\label{R:omega:1}  Note that, whenever monomorphisms are
    \smooth, we have $A_0 = 0$ and the above join $a_j$ is 
    obtained as the colimit of the chain of the subobject
    $a_i\colon A_i \monoto A$, $i < j$ (see \autoref{R:constr}).

  \item\label{R:omega:2} If $F$ is a finitary functor on a locally
    finitely presentable category, then the least ordinal $i$ with
    $a^* = a_i$ is at most $\omega$. Indeed, $\nexttime$ preserves
    joins of $\omega$-chains in $\Sub(A)$ because $F$ does, since
    these joins are obtained as chain colimits
    (see~\cite[Prop.~1.62]{ar}) and so does $\pback \alpha$ since
    colimits of chains are universal
    (cf.~\autoref{E:uconstr}\ref{E:uconstr:4}). By Kleene's fixed
    point theorem $a^* = \bigvee_{i \in \Nat} a_i$.
    
  \item\label{R:omega:3} The same holds for a finitary functor on a category
    with universally \smooth monomorphisms. However, in general
    one needs transfinite iteration to reach a fixed point (see
    \autoref{R:trans}). 
  \end{enumerate}
\end{rem}
\begin{expl}
  Let $(A,\alpha)$ be a graph regarded as a coalgebra for $\Pow$ (see
  \autoref{E:graph}). Then $A_0 = \emptyset$, $A_1$ is formed by all
  leaves, i.e.~those nodes with no neighbours, $A_2$ by all leaves and
  all nodes such that every neighbour is a leaf, etc. We see that a
  node $x$ lies in $A_{i+1}$ iff every path starting in $x$ has length
  at most $i$. Hence $A^* = A_\omega$ is the set of all nodes from
  which no infinite paths starts.
\end{expl}
\begin{notation}
  For every pair $i \leq j$ or ordinals, we denote by
  $a_{ij}\colon A_i \monoto A_j$ the unique morphism witnessing
  $a_i \leq a_j$, i.e.~$a_i = a_j \cdot a_{ij}$. Note that these arise
  by transfinite recursion as well: $a_{0i}$ is obtained by
  initiality, at limit steps use the colimit morphisms, and at
  successor steps one uses the pullback property. That is, in the
  following diagram (in which all vertical morphisms are monomorphisms)
  \begin{equation}\label{521}
    \vcenter{
    \xymatrix@C+2pc{
      A_{i+1} \ar@{ >-->}[d]^-{a_{i+1,j+1}} \ar[r]^{\alpha(a_i)}
      \ar@{ >->}@/_3pc/[dd]_-{a_{i+1}} 
      &  F A_i  \ar@{ >->}[d]_-{Fa_{ij}}  \ar@{ >->}@/^3pc/[dd]^{F a_{i}}    \\
      A_{j+1} \ar@{ >->}[d]^-{a_{j+1}} \ar[r]^{\alpha(a_j)} \pullbackcorner
      &  F A_j  \ar@{ >->}[d]_-{Fa_{j}}    \\
      A \ar[r]_-{\alpha}
      &  F A                 
    }}
  \end{equation}
  the outside commutes by the definitions of $A_{i+1}$,
  $a_{i+1}$, and $\alpha(a_i)$; also the triangle on the right commutes by
  induction hypothesis on $i$.  Since the bottom square is a pullback,
  we obtain $a_{i+1,j+1}$ as desired.
\end{notation}
\smnote{We need to assume existence of colimits of chains!}
\begin{theorem}\label{T:sub}
  Let $\A$ be a complete and well-powered category with universally
  \smooth monomorphisms.  Then for endofunctors preserving
  finite intersections, every subcoalgebra of a well-founded coalgebra
  is well-founded itself.
\end{theorem}
\begin{proof}
  Let $\alpha\colon A\to FA$ be well-founded. Recall the subobjects
  $a_i\colon A_i \monoto A$ from \autoref{C:wf-part}. Let 
  \[
    \xymatrix{
      B \ar[r]^-{\beta} \ar@{ >->}[d]_m 
      &
      FB
      \ar@{ >->}[d]^{Fm}
      \\
      A \ar[r]_-\alpha & FA
    }
  \]
  be a subcoalgebra and denote by $b_i\colon B_i \monoto B$ the
  subobjects of $B$ provided by \autoref{C:wf-part}). There is an
  ordinal $\lambda$ such that $a_\lambda$ is invertible, and we shall
  prove that $b_\lambda$ is also invertible; thus, $(B,\beta)$ is
  well-founded. It is sufficient to prove by transfinite induction
  that the following squares are pullbacks, for suitable monomorphisms $m_i$:
  \[
    \xymatrix{
      B_i \ar[r]^-{b_i} \ar@{ >->}[d]_{m_i}  
      &
      B
      \ar@{ >->}[d]^{m}
      \\
      A_i \ar[r]_-{a_i} & A
    }
  \]
  In other words we prove that for every $i$ we have
  \[
    b_i = \pback m (a_i).
  \]
  For $i = 0$, the statement $b_0 = \pback m (a_0)$ means that the
  square below is a pullback:
  \[
    \xymatrix{
      0 \ar[r]^-{b_0} \ar@{=}[d]
      &
      B
      \ar@{ >->}[d]^{m}
      \\
      0 \ar[r]_-{a_i} & A,
    }
  \]
  which is trivial since $0$ is a strict initial object (see
  \autoref{R:constr}\ref{R:constr:1}).

  For the isolated step we use the induction hypothesis and
  \autoref{C:pback}\ref{C:pback:1} to obtain:
  \[
    b_{i+1}
    =
    \nexttime_\beta(b_i)
    =
    \nexttime_\beta(\pback m (a_i))
    =
    \pback m (\nexttime_\alpha(a_i))
    =
    \pback m (a_{i+1}).
  \]
  For a limit ordinal $j$, we use \autoref{R:constr}\ref{R:constr:3} to obtain
  \[
    b_j
    =
    \bigvee_{i<j} b_i
    =
    \bigvee_{i<j} \pback m (a_i)
    =
    \pback m \big(\bigvee_{i < j} a_i\big)
    =
    \pback m (a_j).
    \tag*{\qed}
  \]
\end{proof}

\section{The General Recursion Theorem}
\label{S:grt}

The main consequence of well-foundedness is parametric
recursivity. This is Taylor's General Recursion
Theorem~\cite[Theorem~6.3.13]{taylor2}. Taylor assumed that $F$
preserves inverse images.  We present a new proof for which it is
sufficient that $F$ preserves monomorphisms, assuming those are
\smooth. In the next section, we discuss the converse implication
in \autoref{T:rec-wf:1} and \autoref{T:rec-wf:2}.
\begin{rem}
  Recall from \autoref{R:ini} the initial-algebra chain for $F$. 
  If $\A$ has \smooth monomorphisms and $F$ preserves
  monomorphisms, then all $w_{i,j}$ in the initial-algebra chain are
  monic. This follows from an easy transfinite induction.
\end{rem}
\begin{theorem}[General Recursion Theorem]  
  \label{T:wf-prec}
  Let $\A$ be a complete and wellpowered category with \smooth
  monomorphisms. For $F\colon \A \to \A$ preserving
  monomorphisms, every well-founded coalgebra is
  parametrically recursive. 
\end{theorem}
\begin{proof}
  \begin{enumerate}
  \item Given an arbitrary coalgebra $(A,\alpha)$ we use the 
    chain of subobjects $a_i\colon A_i \monoto A$ from
    \autoref{C:wf-part} \footnote{One might object to this use of transfinite
   recursion, since Theorem~\ref{T:wf-prec} itself could be used as a 
   justification for transfinite recursion.  Let us emphasize that we are not
    presenting Theorem~\ref{T:wf-prec}
   as a foundational contribution.   We are building on the classical theory of transfinite 
   recursion, extending that result by categorifying it.} We also have the initial-algebra chain
    $W_i = F^i 0$ with connecting morphisms $w_{ji}$ (see
    \autoref{R:ini}). We obtain a natural transformation
    \[
      h_i\colon A_i \to W_i\qquad i \in \Ord,
    \]
    by transfinite recursion as follows: $h_0 =\id_0$, and given
    $h_i\colon A_i \to W_i$, let
    \[
      h_{i+1} = (A_{i+1}
      \xrightarrow{\alpha(a_i)} FA_i
      \xrightarrow{Fh_i} FW_i = W_{i+1}).
    \]
    Finally, for a limit ordinal $i$, $h_i$ is uniquely determined by the
    universal property of the colimit $A_i$.

    We must verify that for $j\leq i$ the naturality square below commutes:
    \begin{equation}\label{eq-shw}
      \vcenter{
        \xymatrix{
          A_j \ar[r]^-{h_j} \ar@{ >->}[d]_{a_{ji}} & W_j \ar@{ >->}[d]^{w_{ji}} \\
          A_i \ar[r]_-{h_i} & W_i
        }
      }
    \end{equation}
    The proof is by transfinite induction on $i$. The base case for
    $i=0$ is trivial, and the step when $i$ is a limit ordinal follows
    from the fact that we use colimits to define both $A_i$ and
    $W_i$. We are left with the successor step $i+1$.  Here we
    again use transfinite induction on $j$.  The verification amounts to
    assuming~\eqref{eq-shw} for $i$ and $j$ and showing the same
    equation for $j+1$ and $i+1$.  For this, consider the diagram
    below:
    \begin{equation}\label{2traps}
      \vcenter{
        \xymatrix@C+1pc{
          A_{j+1} \ar[dd]_-{h_{j+1}} \ar[rrr]^{a_{j+1,i+1}} \ar[dr]^{\alpha(a_j)} 
          & & &
          A_{i+1}  \ar[dd]^-{h_{i+1}}   \ar[dl]_-{\alpha(a_i)}
          \\
          &
          FA_{j+1}\ar[dl]_-{F h_j}  \ar[r]^-{F a_{ji}}
          &
          FA_{i+1} \ar[dr]^-{F h_i} 
          \\
          \llap{$W_{j+1} =\ $}
          F^{j+1}0 \ar[rrr]_-{w_{j+1,i+1}}
          & & &
          F^{i+1}0\rlap{$\ = W_{i+1}$}
        }}
    \end{equation}
    The region at the top is also the top square of~\eqref{521}, the
    triangles commute by the definition of $(h_i)$, and the region at the bottom
    commutes by the induction hypothesis and the fact that
    $F w_{ji} = w_{j+1,i+1}$.  Thus the outside commutes, as desired.
    
  \item\label{T:wf-prec:2} Now suppose that $(A,\alpha)$ is a
    well-founded coalgebra. We prove that $(A,\alpha)$ is
    recursive, i.e.~for every algebra $e\colon FX \to X$ we present a
    coalgebra-to-algebra morphism $\sol e$ and prove that it is
    unique.

    For every ordinal $i$, the coalgebra $w_{i,i+1}\colon W_i \to FW_i$
    is recursive (see \autoref{E:reco}\ref{E:reco:6}). Hence we have a
    morphism $f_i\colon W_i\to X$ such that the square on the bottom
    below commutes:
    \begin{equation}\label{D:some}
      \vcenter{
        \xymatrix@C+.5pc{
          A \ar[r]^-{\alpha(a_i) = \alpha} \ar[d]_-{h_i}
          \ar[dr]^-{h_{i+1}}
          &
          FA  \ar[d]^-{Fh_i}
          \\
          W_i  \ar[d]_-{f_i} \ar[r]_-{w_{i,i+1}}
          &
          FW_i  \ar[d]^-{F f_i}
          \\
          X & \ar[l]^-{e} FX 
        }}
    \end{equation}
    Since $(A,\alpha)$ is well-founded, there exists an ordinal $i$
    such that $A = A_i = A_{i+1}$ (see \autoref{C:wf-part}). Then
    we have $\alpha(a_i) = \alpha$, so that the upper triangle commutes
    by definition of $h_{i+1}$. Moreover, the lower triangle is an instance
    of~\eqref{eq-shw} using the fact that $a_{i, i+1}
    =\id$. Thus the outside of the diagram commutes, and so
    $f_i\o h_i$ is the desired coalgebra-to-algebra morphism.

  \item For the uniqueness, suppose that $\sol e$ is any
    coalgebra-to-algebra morphism from $\alpha$ to $e$, i.e.~in the
    diagram below the lower square commutes:
    \begin{equation}\label{D:uniq}
      \vcenter{
        \xymatrix{
          A_{i+1}
          \ar[r]^-{\alpha(a_i)}
          \ar@{ >->}[d]_-{a_{i+1}}
          &
          FA_i
          \ar@{ >->}[d]^-{Fa_i}
          \\
          A
          \ar[d]_-{\sol e}
          \ar[r]^-{\alpha}
          &
          FA
          \ar[d]^-{F\sol e}
          \\
          X & \ar[l]^-{e} FX 
        }}
    \end{equation}
    Moreover, the upper one is the square defining $\alpha(a_i)$ (see
    \autoref{D-tilde}). 

    We verify by induction on $j$ that
    $e^\dag \o a_j = f_j \o h_j \o a_j$. Then for the above ordinal
    $i$ with $a_i = \id_A$, we have $e^\dag = f_i \o h_i$ as
    desired. For the base case $j=0$, the equation trivially holds,
    and for limit ordinals $j$ we use the universal property of the
    colimit $A_j$. For the successor step we use that~\eqref{D:uniq}
    and~\eqref{D:some} commute (with $j$ substituted for $i$). By
    pasting~\eqref{D:some} and the upper square
    of~\eqref{D:uniq} we obtain
    \begin{equation}\label{eq:paste}
      e \o F(f_j\o h_j \o a_j) \o \alpha(a_j) = f_j \o h_j\o a_{j+1}
    \end{equation}
    This yields the desired equality:
    \begin{align*}
      \sol e \o a_{j+1} &= e \o F(\sol e \o a_j) \o \alpha(a_j)
      & \text{(by~\eqref{D:uniq})} \\
      &= e \o F(f_j\o h_j \o a_j) \o \alpha(a_j)
      & \text{(by induction hypothesis)}
      \\
      &= f_j \o h_j\o a_{j+1}
      &\text{(by~\eqref{eq:paste}).}
    \end{align*}

  \item Finally, we prove that the coalgebra $(A,\alpha)$ is a
    parametrically recursive. 
    
    Consider the
    coalgebra $\pair{\alpha,\id_A}\colon A \to FA \times A$ for $F(-)
    \times A$. This functor preserves monomorphisms since $F$ does
    and monomorphisms are closed under products. The next
    time operator $\nexttime$ on $\Sub(A)$ is the same for both coalgebras since the
    square~\eqref{rdj} is a pullback if and only if the square
    below is one:
    \[
      \begin{tikzcd}[column sep = 20mm]
        \nexttime S
        \pullbackangle{-45}
        \arrow{r}{\pair{\alpha(s),\nexttime s}}
        \arrow[>->]{d}[swap]{\nexttime s}
        &
        FS \times A
        \arrow[>->]{d}{Fs \times A}
        \\
        A
        \arrow{r}{\pair{\alpha, A}}
        &
        FA \times A
      \end{tikzcd}
    \]
    Since $\id_A$ is the unique fixed point of $\nexttime$ w.r.t.~$F$,
    it is also the unique fixed point of $\nexttime$
    w.r.t.~$F(-) \times A$. Thus, $(A,\pair{\alpha,\id_A})$ is a
    well-founded coalgebra for $F(-)\times A$. By
    point~\ref{T:wf-prec:2}, it is thus recursive for $F(-) \times
    A$. This states equivalently that $(A,\alpha)$ is a parametrically
    recursive coalgebra for $F$.\qed
  \end{enumerate}
\end{proof}
\begin{corollary}\label{C:wf-prec}
  For every endofunctor on $\Set$ or $\KVec$ (vector spaces and linear
  maps), every well-founded coalgebra is parametrically recursive.
\end{corollary}
\begin{proof}
  For $\Set$, we apply \autoref{T:wf-prec} to the Trnkov\'a hull
  $\bar F$ (see \autoref{P:Tr}), noting that $F$ and $\bar F$ have the
  same (non-empty) coalgebras. By \autoref{L:Trn} the desired result follows.
  For $\KVec$, observe that monomorphisms split and are therefore
  preserved by every endofunctor $F$.
  \qed
\end{proof}
\begin{expl}\label{E:wf-prec}
  For the set functor $FX = X \times X  + 1$ the coalgebra
  $(\Nat,\gamma)$ from \autoref{E:cangr}\ref{E:cangr:3} is
  well-founded. Hence it is parametrically recursive.

  Similarly, we saw that for $FX = 1 + A \times X \times X$ the
  coalgebra $(A,\qsplit)$ from \autoref{E:prec}\ref{E:prec:3} is
  well-founded, and therefore it is  (parametrically) recursive. 
\end{expl}
\begin{expl}
  Well-founded coalgebras need not be recursive when $F$ does not
  preserve monomorphisms.  We take $\A$ to be the category of
  \emph{sets with a predicate}, i.e.~pairs $(X,A)$, where
  $A\subseteq X$.  Morphisms $f\colon (X,A) \to (Y,B)$ satisfy
  $f[A]\subseteq B$. Denote by $\mathbb{1}$ the terminal object
  $(1,1)$.  We define an endofunctor $F$ by
  $F(X,\emptyset) = (X+1,\emptyset)$, and for $A\neq\emptyset$,
  $F(X,A) = \mathbb{1}$.  For a morphism $f\colon (X,A)\to (Y,B)$, put
  $F = f + \id$ if $A = \emptyset$; if $A \neq \emptyset$, then
  also $B \neq \emptyset$ and $Ff$ is
  $\id\colon \mathbb{1}\to \mathbb{1}$.

  The terminal coalgebra is $\id\colon \mathbb{1}\to \mathbb{1}$, and
  it is easy to see that it is well-founded. But it is not
  recursive: there are no coalgebra-to-algebra morphisms into an
  algebra of the form $F(X,\emptyset) \to (X,\emptyset)$.
\end{expl}

We close with a general fact on well-founded parts of \emph{fixed
  points} (i.e.~(co)al\-ge\-bras whose structure is invertible). The
following result generalizes~\cite[Cor.~3.4]{JeanninEA17}, and it also
appeared before for functors preserving finite
intersections~\cite[Theorem~8.16 and Remark~8.18]{amm18}.
Here we lift the latter assumption:
\begin{theorem}
  Let $\A$ be a complete and well-powered category with \smooth
  monomorphisms. For $F$ preserving monomorphisms, the well-founded part of
  every fixed point is an initial algebra. In particular, the only
  well-founded fixed point is the initial algebra.
\end{theorem}
\begin{proof}
  Let $\alpha\colon A \to FA$ be a fixed point of $F$. By
  \autoref{R:ini}\ref{R:ini:2} we know that the initial algebra $(\mu
  F, \ini)$ exists. Now let $a^*\colon (A^*, \alpha^*) \monoto (A, \alpha)$ be
  the well-founded part of $A$ given in \autoref{P:wfdpart}. This is a
  cartesian subcoalgebra, i.e.~we have a pullback square
  \[
    \begin{tikzcd}
      A^* \arrow[>->]{d}[swap]{a^*}
      \arrow{r}{\alpha^*}
      \pullbackangle{-45}
      &
      FA^*
      \arrow[>->]{d}{Fa^*}
      \\
      A
      \arrow{r}{\alpha}
      &
      FA
    \end{tikzcd}
  \]
  Since $\alpha$ is an isomorphism, so is $\alpha^*$. 

  By initiality, we have an algebra homomorphism $h\colon (\mu F,
  \ini) \to (A^*, (\alpha^*)^{-1})$, i.e.~a coalgebra homomorphism
  \[
    \begin{tikzcd}
      \mu F \arrow{r}{\ini^{-1}} \arrow{d}[swap]{h}
      \pullbackangle{-45}
      &
      F(\mu F) \arrow{d}{Fh}
      \\
      A^* \arrow{r}{\alpha^*} &
      FA^*      
    \end{tikzcd}
  \]
  Since both horizontal morphisms are invertible, this square is a
  pullback. By \autoref{T:wf-prec}, $(A^*, \alpha^*)$ is
  recursive. Thus, we have a coalgebra homomorphism
  $k\colon (A^*,\alpha^*) \to (\mu F, \ini^{-1})$ by
  \autoref{cor:cuv}. By the universal property of $\mu F$, we obtain
  $k \cdot h = \id_{\mu F}$, whence $h$ is a split monomorphism. Thus
  the above square exhibits $(\mu F, \ini^{-1})$ as a cartesian
  subcoalgebra of $(A^*, \alpha^*)$. By
  \autoref{R:fixed}\ref{R:fixed:1}, we conclude that $h$ is an
  isomorphism.
  \qed
\end{proof}
\begin{expl}
  We illustrate that for a set functor $F$
  preserving monomorphisms, the well-founded part of the terminal coalgebra is
  the initial algebra.  Consider $FX = A\times X + 1$.  The terminal
  coalgebra is the set $A^\infty \cup A^*$ of finite and infinite
  sequences from the set $A$.  The initial algebra is $A^*$.  It is
  easy to check that $A^*$ is the well-founded part of
  $A^\infty \cup A^*$.
\end{expl}

\section{The Converse of the General Recursion Theorem}

We prove a converse to \autoref{T:wf-prec}:    
\( \text{``recursive $\implies$ well-founded''.}  \)  
 Related results appear in
Taylor~\cite{taylor3,taylor2}, Ad\'amek et al.~\cite{alm_rec} and
Jeannin et al.~\cite{JeanninEA17}.  

For this, one needs to assume more than preservation of finite intersections. In
fact, we will assume that $F$ preserves inverse images. But even this is not
enough. We additionally assume that either
\begin{enumerate}
\item The underlying category $\A$ has universally \smooth
  monomorphisms and the endofunctor $F$ has a pre-fixed point (see
  \autoref{R:ini}\ref{R:ini:2}). 
  \item The underlying category $\A$ has a subobject classifier.
\end{enumerate}

The first of these possible assumptions leads to
Theorem~\ref{T:rec-wf:1}, the second is a theorem of
Taylor~\cite{taylor3}. Finally, at the end of this section we prove
the above converse implication for every functor on vector spaces
preserving inverse images (see \autoref{T:recwfvec}). This last result is not
covered by the previous two results since $\KVec$ neither has
universally constructive monomorphims nor a subobject classifier.

\smnote{We need to assume existence of colimits of chains!}
\begin{theorem}\label{T:rec-wf:1}
  Let $\A$ be a complete and wellpowered category with universally
  \smooth monomorphisms, and suppose that $F\colon \A \to \A$ preserves
  inverse images and has a pre-fixed point. Then every recursive
  $F$-coalgebra is well-founded. 
\end{theorem}
\begin{proof}
  First observe that an intial algebra exists by
  \autoref{R:ini}\ref{R:ini:2}. 
  Now suppose that $(A,\alpha)$ is a recursive coalgebra. Then there
  exists a unique coalgebra homomorphism
  $h\colon (A,\alpha) \to (\mu F, \iota^{-1})$. Let us abbreviate
  $w_{i\lambda}$ by $c_i \colon F^i 0 \monoto \mu F$ and recall the
  subobjects $a_i\colon A_i \monoto A$ from \autoref{C:wf-part}. We
  are going to prove by transfinite induction that for every
  $i \in \Ord$, $a_i$ is the inverse image of $c_i$ under $h$, i.e.~we
  have a pullback square
  \begin{equation}\label{eq:invim}
    \begin{tikzcd}
      A_i
      \arrow[>->]{d}[swap]{a_i}
      \arrow{r}{h_i}
      \pullbackangle{-45}
      &
      W_i
      \arrow[>->]{d}{c_i}
      \\
      A\arrow{r}{h}
      &
      \mu F
    \end{tikzcd}
    \qquad\text{for some morphism $h_i\colon A_i \to W_i$};
  \end{equation}
  in symbols: $a_i = \pback h(c_i)$ for all ordinals $i$. Then it
  follows that $a_\lambda$ is an isomorphism, since so is
  $c_\lambda$, whence $(A,\alpha)$ is well-founded. In the base case
  $i =0$ the above square clearly is a pullback since $A_0 = W_0 = 0$ is a
  strict initial object (see \autoref{R:omega}\ref{R:omega:1}). 
  
  For the isolated step we compute the pullback of
  $c_{i+1}\colon W_{i+1} \to \mu F$ along $h$ using the following
  diagram:
  \[
    \begin{tikzcd}
      A_{i+1}
      \pullbackangle{-45}
      \arrow[>->]{d}[swap]{a_{i+1}}
      \arrow{r}{\alpha(a_i)}
      &
      FA_i
      \pullbackangle{-45}
      \arrow[>->]{d}[swap]{Fa_i}
      \arrow{r}{Fh_i}
      &
      FW_i
      \arrow[>->]{d}[swap]{Fc_i}
      \arrow{rd}{c_{i+1}}
      \\
      A
      \arrow{r}{\alpha}
      \arrow[shiftarr={yshift=-20pt}]{rrr}{h}
      &
      FA
      \arrow{r}{Fh}
      &
      F(\mu F)
      \arrow{r}{\ini}
      &
      \mu F
    \end{tikzcd}
  \]
  By the induction hypothesis and since $F$ preserves inverse
  images, the middle square above is a pullback. Since the structure
  map $\ini$ of the initial algebra is an isomorphism, it follows
  that the middle square pasted with the right-hand triangle is also
  a pullback. Finally, the left-hand square is a pullback by the
  definition of $a_{i+1}$. Thus, the outside of the above diagram is
  a pullback, as required.
  
  For a limit ordinal $j$, we know that $a_j = \bigvee_{i<j} a_i$
  and similarly, $c_j = \bigvee_{i< j} c_i$ since
  $W_j = \colim_{i<j}W_j$ and monomorphisms are \smooth (see
  \autoref{R:constr}\ref{R:constr:2}). Using
  \autoref{R:constr}\ref{R:constr:3} and the induction hypothesis we
  thus obtain
  \[
    \pback h (c_j)
    =
    \pback h \big(\bigvee_{i < j} c_i\big)
    =
    \bigvee_{i<j} \pback h (c_i)
    =
    \bigvee_{i<j} a_i
    =
    a_j.
    \tag*{\qed}
  \]
\end{proof}

\begin{corollary}\label{C:equiv}
  Let $\A$ and $F$ satisfy the assumptions of \autoref{T:rec-wf:1}.
  Then the following properties of a coalgebra are equivalent:
  \begin{enumerate}
  \item well-foundedness,
  \item parametric recursiveness,
  \item recursiveness,
  \item existence of a homomorphism into $(\mu F,
    \ini^{-1})$,
  \item existence of a homomorphism into a well-founded one.
  \end{enumerate}
\end{corollary}
\begin{proof}
  We already know (1) $\Rightarrow$ (2) $\Rightarrow$ (3). Since $F$
  has an initial algebra (as proved in \autoref{T:rec-wf:1}), the
  implication (3) $\Rightarrow$ (4) follows from \autoref{cor:cuv}. In
  \autoref{T:rec-wf:1} we also proved (4) $\Rightarrow$ (1). The implication (4)
  $\Rightarrow$ (5) follows from
  \autoref{E-well-founded}\ref{E-well-founded:1}. Finally, it follows
  from~\cite[Remark~2.40]{amms} that
  $(\mu F, \ini^{-1})$ is a terminal well-founded coalgebra.  Thus,
  (5) $\Rightarrow$ (4), which completes the proof.\qed
\end{proof}
\takeout{
\begin{expl}\label{E:vec}
  The identity functor on $\KVec$, the category of vector spaces over
  the field $K$, has recursive coalgebras that are not
  well-founded. For example, $0\colon K \to K$ is a recursive
  coalgebra. Indeed, given an algebra $\alpha\colon A \to A$ the
  morphism $h = 0$ is obviously the unique morphism making the
  following square commutative:
  \[
    \begin{tikzcd}
      K
      \arrow{d}[swap]{h}
      \arrow{r}{0}
      &
      K
      \arrow{d}{h}
      \\
      A
      &
      A
      \arrow{l}{\alpha}
    \end{tikzcd}
  \]
  However, $(K,0)$ is not well-founded since it has $(0,0)$ as a
  cartesian subcoalgebra.

  This is a bit surprising since $\KVec$ ``almost'' has universally
  \smooth monomorphisms: it is just that the initial object is
  not strict. However, for all ordinals $\lambda \neq 0$ colimits of
  $\lambda$-chains are universal. 
\end{expl}}
\begin{expl}
  \begin{enumerate}
  \item The category of many-sorted sets satisfies the assumptions of
    \autoref{T:rec-wf:1}, and polynomial endofunctors on that category
    preserve inverse images. Thus, we obtain Jeannin et al.'s
    result~\cite[Thm.~3.3]{JeanninEA17} that (1)--(4) in
    \autoref{C:equiv} are equivalent as a special instance.

  \item Recall from \autoref{E:uconstr}\ref{E:uconstr:2} that vector
    spaces fail to have universally smooth monomorphisms. The
    implication (4) $\Rightarrow$ (3) in \autoref{C:equiv} does not
    hold for vector spaces. In fact, for the identity functor on
    $\KVec$ we have $\mu \Id = (0, \id)$. Hence, every coalgebra has a
    homomorphism into $\mu \Id$. However, not every coalgebra is
    recursive, e.g.~the coalgebra $(K, \id)$ admits many
    coalgebra-to-algebra morphisms to the algebra
    $(K,\id)$. Similarly, the implication (4) $\Rightarrow$ (1) does
    not hold. In fact, a coalgebra $\alpha\colon A \to A$ is
    well-founded iff for every $x \in A$ there exists a natural number
    $n$ with $\alpha^n(x) = 0$
    (cf.~\autoref{E-well-founded}\ref{E-well-founded:4}). Clearly, not
    every coalgebra satisfies this property. In contrast, see \autoref{C:vec}. 
  \end{enumerate}
\end{expl}
\begin{rem}\label{R:trans}
  Coming back to \autoref{R:omega}\ref{R:omega:3}, we see from the
  proof of \autoref{T:rec-wf:1} that in general one needs transfinite
  iteration to obtain the least fixed point $\nexttime$. Indeed, for
  $(A, \alpha) = (\mu F, \ini^{-1})$ we have $h = \id$ in
  \eqref{eq:invim} and therefore $a_i = c_i$. Now for
  $FX = X^\Nat + 1$ on $\Set$ we have that $\mu F$ is carried by the
  set of all (ordered) well-founded countably-branching
  trees. Furthermore, it is easy to show that $\mu F = W_{\omega_1}$,
  where $\omega_1$ is the first uncoutable ordinal, and each $W_i$,
  $i < \omega_1$ is a proper subset.
\end{rem}

In \autoref{T:rec-wf:1}, we assumed that the endofunctor has a
pre-fixed point.  For set functors, this assumption may be lifted.
Indeed, whenever a category has a subobject classifier, then every
recursive coalgebra is well-founded, as shown by
Taylor~\cite[Rem.~3.8]{taylor3}. We present this in all details for
convenience of the reader.

\begin{rem}
  \begin{enumerate}
  \item Let us recall the definition of a \emph{subobject classifier}
    originating in~\cite{Lawvere70} and prominent in topos theory.
    This is an object $\Omega$ with a subobject $t\colon 1 \to \Omega$
    such that for every subobject $b\colon B\monoto A$ there is a
    unique $\hat b\colon C\to \Omega$ such that the square below is a
    pullback:
    \begin{equation}\label{Om}
      \vcenter{
        \xymatrix{
          B \ar[r]^-{!} \ar@{ >->}[d]_-b  \pullbackcorner
          &
          1\ar@{ >->}[d]^{t}
          \\
          A \ar[r]_-{\hat b}
          &
          \Omega		
        }}
    \end{equation}
    By definition, every elementary topos has a subjobject classifier,
    in particular every category $\Set^\C$ with $\C$ small.

  \item $\Set$ has a subobject classifier given by $\Omega =
    \set{t,f}$ with the evident $t\colon 1 \subto \Omega$. Indeed, subsets
    $b\colon B \subto A$ are in one-to-one correspondence with
    characteristic maps $\hat b\colon B \to \Omega$.

  \item Our standing assumption that $\A$ is a complete and
    well-powered category is not needed for the next result: finite
    limits are sufficient.
  \end{enumerate}
\end{rem}
\begin{theorem}[{Taylor~\cite{taylor3}}]\label{T:rec-wf:2}
  Let $F$ be an endofunctor preserving inverse images on a finitely
  complete category with a subobject classifier. Then every recursive
  $F$-coalgebra is well-founded.
\end{theorem}
\begin{proof}
  Let $(A,\alpha)$ be a recursive coalgebra. Clearly, $\id_A$ is a
  fixed point of $\nexttime$, and we prove below that it is the unique
  one. Thus, $(A,\alpha)$ is well-founded.

  Let $b\colon B \to FB$ be any fixed point of $\nexttime$. Consider
  the following diagram:
  \[    
    \begin{tikzcd}
      B
      \pullbackangle{-45}
      \arrow{r}{\alpha(b)}
      \arrow[>->]{d}[swap]{b}
      \arrow[shiftarr = {yshift = 5mm}]{rrr}{!}
      &
      FB
      \pullbackangle{-45}
      \arrow{r}{F!}
      \arrow[>->]{d}[swap]{Fb}
      &
      F1
      \pullbackangle{-45}
      \arrow{r}{!}
      \arrow[>->]{d}[swap]{Ft}
      &
      1
      \arrow[>->]{d}{t}
      \\
      A
      \arrow{r}{\alpha}
      \arrow[shiftarr = {yshift = -5mm}]{rrr}{\widehat b}
      &
      FA
      \arrow{r}{F\widehat b}
      &
      F\Omega
      \arrow{r}{\widehat{Ft}}
      &
      \Omega
    \end{tikzcd}
  \]
  The square on the left is a pullback because $b = \nexttime b$.  The
  central square is $F$ applied to the pullback square (\ref{Om}) for
  $b\colon B\monoto A$.  The square on the right is the pullback square
  (\ref{Om}) for $Ft\colon F1 \to F\Omega$.  The upper morphism is
  $!\colon B\to 1$, and so the lower one is $\widehat{b}$.  Thus the
  outside rectangle is again a pullback.  In particular,
  \[
    \widehat{b} =  \widehat{Ft} \o F\widehat{b} \o \alpha.
  \]
  So we have a coalgebra-to-algebra morphism
  \[
    \xymatrix{
      A \ar[r]^-{\alpha} \ar[d]_-{\widehat{b}}  
      &
      FA\ar[d]^{F\widehat{b}}
      \\
      \Omega	
      & \ar[l]^-{\widehat{Ft}}
      F\Omega		
    }
  \]
  Since $(A,\alpha)$ is recursive, this means that $\widehat{b}$ is
  uniquely determined by $\alpha$, independent of which fixed point $b$
  of $\nexttime$ was used in our argument. Thus  $\widehat{b} =
  \widehat{\id_A}$, as desired.\qed
\end{proof}
\takeout{
\begin{corollary}
  For every set functor preserving inverse images, recursive coalgebras
  are well-founded. 
\end{corollary}
\begin{corollary}\label{C:equiv:2}
  Let $\A$ and $F$ satisfy the assumptions of \autoref{T:rec-wf:2}.
  The the following properties of a coalgebra are 
  \[
    \mbox{well-foundedness} \iff
    \mbox{recursiveness} \iff
    \mbox{parametrically recursiveness}.
  \]
\end{corollary}
}
\begin{corollary}\label{C:equiv:2}
  For every set functor preserving inverse images, the following
  properties of a coalgebra are equivalent:
  \[
    \mbox{well-foundedness} \iff
    \mbox{parametric recursiveness} \iff
    \mbox{recursiveness}.
  \]
\end{corollary}

%
\begin{expl}\label{exSq0}
  The hypothesis in \autoref{T:rec-wf:1} and \autoref{T:rec-wf:2} that
  the functor preserves inverse images cannot be lifted. In order to
  see this, we consider the functor $R\colon \Set \to \Set$ of
  \autoref{E:setfunctors}\ref{E:setfunctors:2}. It preserves
  monomorphisms but not inverse images. The recursive coalgebra
  $(C, \gamma)$ in \autoref{E:functorR} is not well-founded:
  $\emptyset$~is a cartesian subcoalgebra.
\end{expl}

We have seen that for set functors well-founded coalgebras are
recursive, and the converse holds for functors preserving inverse
images. Moreover, the latter requirement cannot be lifted as we just
saw in \autoref{exSq0}. Recall that an initial algebra $(\mu F, \ini)$
is also considered as a coalgebra $(\mu F, \ini^{-1})$.
Taylor~\cite[Cor.~9.9]{taylor3} showed that, for functors preserving
inverse images, the terminal well-founded coalgebra is the initial
algebra. Surprisingly, this result is true for \emph{all} set
functors.
\begin{notheorembrackets}%
\begin{theorem}[{\cite[Thm.~2.46]{amms}}]\label{T-in-ter-sets} 
  For every set functor, a terminal well-founded coalgebra is
  precisely an initial algebra. 
\end{theorem}
\end{notheorembrackets}
The proof is nontrivial, and we are not going to present it. It is
based on properties of well-founded coalgebras in locally presentable
categories. The fact that no assumptions on $F$ are needed seems very
special to $\Set$. On the one hand, \autoref{T-in-ter-sets} can be
proved for every locally finitely presentable base category $\A$
having a strict initial object and every endofunctor on $\A$
preserving finite intersections~\cite[Theorem~2.36]{amms}. On the
other hand, without this last assumptions, \autoref{T-in-ter-sets}
does not even generalize from $\Set$ to the category of graphs as the
following example shows.
	
\begin{expl}
  Let $\Gra$ be the category of graphs, i.e.~the category of
  presheaves over the category
  $\{ \bullet \rightrightarrows \bullet\}$ given by two parallel
  morphisms.  Here is a simple endofunctor $F$ on $\Gra$ whose initial
  algebra is infinite and whose terminal well-founded coalgebra is a
  singleton graph: On objects $A$ put $FA = 1$ (the terminal graph) if
  $A$ has edges. For a graph $A$ without edges, let $FA$ be the graph
  $A + 1$ without edges. The definition of $F$ on morphisms
  $h\colon A \to B$ is as expected: $Fh$ maps the additional vertex of
  $A$ to that of $B$ in the case where $B$ has no edges. Then $\mu F$
  is the graph of natural numbers without edges. However, the terminal
  well-founded coalgebra is $F1 \xrightarrow{\cong} 1$.
\end{expl}


As the last result of this section we now turn to showing the
implication ``recursive $\Rightarrow$ well-founded'' for functors on the
category $\KVec$ preserving inverse images. This follows neither from
either \autoref{T:rec-wf:1} (since monomorphism are not universally
\smooth in $\KVec$) nor from \autoref{T:rec-wf:2} (since $\KVec$ does
not have a subobject classifier). 

Recall first that the \emph{kernel} of a linear map $f\colon X \to Y$ is the
subspace $\ker f = \set{x \in X : f(x) = 0}$. A functor $F\colon \KVec
\to \KVec$ \emph{preserves kernels} if for every linear map $f\colon X
\to Y$ its kernel $s\colon \ker f \monoto X$ is mapped to the kernel
of $Ff$, shortly $Fs = \ker Ff$. 
\begin{rem}\label{R:vec}
  \begin{enumerate}
  \item\label{R:vec:1}\sloppypar Observe that for every linear map
    $f\colon X \to Y$ its kernel $s\colon \ker f \monoto X$ is the
    inverse image of the least subobject $z\colon 0 \monoto Y$,
    shortly $\ker f = f^{-1}[0]$.

    If $F\colon \KVec \to \KVec$ preserves inverse images and $F0=0$,
    then it preserves kernels. Indeed, $Fs$ is then the inverse image
    of $Fz$ under $Ff$, and $Fz\colon 0 = F0 \monoto FY$ is the zero
    map. Thus $Fs$ is the kernel of $Ff$ as desired.

  \item\label{R:vec:2} Conversely, if $F$ preserves kernels, then $F0 = 0$ (the
    terminal object) and $F$ preserves inverse images. In fact, $F$
    preserves finite limits: by~\cite[Thm.~3.12]{Freyd64}, a functor
    preserving kernels is additive, and for an additive functor
    preservation of kernels is equivalent to preservation of finite
    limits (see~\cite[Prop.~1.11.2]{Borceux94_2}). 

  \item Every subspace $s\colon S \monoto X$ induces a quotient space
    $X/S$, and we denote the corresponding canonical quotient map by
    $\coker s\colon S \to X/S$.
    
  \item\label{R:vec:3} Every linear map $f\colon X \to Y$ induces an isomorphism $X
    \cong \ker f + f[X]$, where $f[X]$ denotes the image of $f$ in $Y$.

  \item\label{R:vec:3} For a linear map $f\colon X \to Y$ and a
    subspace $s\colon S \monoto Y$ let
    $t\colon T = f^{-1}[S] \monoto Y$. Then there exists a unique
    monomorphism $u\colon X/T \monoto Y/S$ such that the following
    diagram commutes:
    \[
      \begin{tikzcd}
        T
        \arrow{r}
        \arrow[>->]{d}[swap]{t}
        \pullbackangle{-45}
        &
        S\arrow[>->]{d}{s}
        \\
        X
        \arrow[->>]{d}[swap]{\coker t}
        \arrow{r}{f}
        &
        Y
        \arrow[->>]{d}{\coker s}
        \\
        X/T
        \arrow[dashed,>->]{r}{u}
        &
        Y/S
      \end{tikzcd}
    \]
    Indeed, $u$ exists by the universal property of $\coker
    t$. Moreover, we see that $u$ is injective: if $x + T$
    satisfies $u(x + T) = 0$, i.e.~$f(x) + S = 0$, then we have $f(x)
    \in S$, thus $x \in T$.
        
\takeout{
  \item Given a linear map $f\colon X \to Y$ and a subspace
    $s\colon S \monoto X$, we have for every vector space $R$ that $s
    = \ker f$ iff $R+s\colon R+S \monoto R+X$ is the kernel of
    $R+f\colon R+X \to R+Y$. Thus follows from the fact that $+$ is
    also the product, and products commute with equalizers.
    \smerror[inline]{Unfortunately, this does not work! $s\colon \ker f \monoto X$
      is the equalizer of $f$ and $0$ which implies that $R\times s$
      is the equalizer of $R \times f$ and $R \times 0$, but the
      latter morphism is not the zero morphism from $R\times X \to R
      \times Y$. In fact, the kernel of $R \times f$ is $\langle 0,
      s\rangle \ker f \monoto R \times X$.}}
  \end{enumerate}
\end{rem}

\begin{theorem}\label{T:recwfvec}
  Let $F$ be an endofunctor on $\KVec$ preserving inverse images. Then every
  recursive $F$-coalgebra is well-founded. 
\end{theorem}
\begin{proof}
  Let $\alpha\colon A \to FA$ be a recursive coalgebra and let
  $a^*\colon (A^*, \alpha^*) \monoto (A,\alpha)$ be its well-founded part.
  \begin{enumerate}
  \item\label{T:recwfvec:1} Assume first that $F0 = 0$. Then $F$
    preserves zero maps and kernels by
    \autoref{R:vec}\ref{R:vec:1}. Then $Fa^*$ is the kernel of
    $F(\coker a^*)$ as shown in the following diagram:
    \[
      \begin{tikzcd}
        A^*
        \arrow{r}{\alpha^*}
        \arrow[>->]{d}[swap]{a^*}
        \pullbackangle{-45}
        &
        FA^*\arrow[>->]{d}{Fa^*}
        \\
        A
        \arrow[->>]{d}[swap]{\coker a^*}
        \arrow{r}{\alpha}
        &
        FA
        \arrow[->>]{d}{F(\coker a^*)}
        \\
        A/A^*
        \arrow[dashed,>->,yshift=2pt]{r}{u}
        &
        F(A/A^*)
        \arrow[dashed,->>,yshift=-2pt]{l}{e}
      \end{tikzcd}
    \]
    Since $\coker a^*$ is epimorphic so is $F(\coker a^*)$ since 
    epimorphisms split in $\KVec$. Thus, we have $F(\coker a^*) =
    \coker (Fa^*)$, and by \autoref{R:vec}\ref{R:vec:3} we obtain
    the unique monomorphism $u\colon A/A^* \monoto F(A/A^*)$ such that
    the diagram above commutes. Choose a splitting $e\colon F(A/A^*)
    \epito A/A^*$, i.e.~$e \cdot u = \id$. It follows that $q = \coker
    a^*$ is a coalgebra-to-algebra morphism from $(A,\alpha)$ to
    $(A/A^*, e)$. Indeed, we obtain
    \[
      e \cdot Fq \cdot \alpha = e \cdot u \cdot q = q.
    \]
    Since $F$ preserves zero morphisms, the zero morphism
    $z\colon A \to A/A^*$ is also a coalgebra-to-algebra
    morphism. Consequently, $q = z$, which is equivalent to $a^*$ being
    an isomorphism $A \cong A^*$ as desired.
    
  \item Let $F$ be arbitrary, and put $R = F0$. Then there is an
    endofunctor $G$ on $\KVec$ with $G0 = 0$ and preserving inverse
    images such that $FX = R \times GX$. Indeed, for every vector space
    $X$, let $t_X\colon X \to 0$ denote the zero map, and let
    $k_X\colon GX \monoto FX$ be the kernel of $Ft_X$. For every linear
    map $f\colon X \to Y$ the equality $t_X = t_Y \cdot f$ implies
    that $Ff$ yields a linear map $Gf$ making the following square
    commutative:
    \[
      \begin{tikzcd}
        GX
        \arrow[>->]{r}{k_X}
        \arrow{d}[swap]{Gf}
        &
        FX
        \arrow{d}{Ff}
        \\
        GY
        \arrow[>->]{r}{k_Y}
        &
        FY
      \end{tikzcd}
    \]
    It is easy to verify that $G$ is an endofunctor and
    $k\colon G\monoto F$ a natural transformation. Observe that $t_X$
    is a split epimorphism (whose splitting is the unique
    $s_X\colon 0 \to X$), whence $Ft_X$ is a split epimorphism with
    splitting $Fs_X\colon R \to FX$. Using
    \autoref{R:vec}\ref{R:vec:3}, this implies that $FX \cong R + GX$
    with coproduct injections $Fs_X$ and $k_X$. Since $+$ is also
    product, we obtain $FX \cong R \times GX$ as desired.

  \item We prove that $G$ preserves kernels. By
    \autoref{R:vec}\ref{R:vec:2}, $G$ then preserves finite limits,
    whence inverse images. Suppose that $s = \ker f$ so that we have
    the pullback on the left below
    \[
      \begin{tikzcd}
        S
        \arrow{r}
        \arrow{d}[swap]{s}
        \pullbackangle{-45}
        &
        0
        \arrow{d}
        \\
        X
        \arrow{r}{f}
        &
        Y
      \end{tikzcd}
      \qquad\qquad
      \begin{tikzcd}
        GS
        \arrow{r}
        \arrow{d}[swap]{Gs}
        \pullbackangle{-45}
        &
        0
        \arrow{d}
        \\
        GX
        \arrow{r}{Gf}
        &
        GY
      \end{tikzcd}
    \]
    It is our task to prove that the square on the right above is a
    pullback. Since $F$ preserves inverse images, applying it to
    left-hand square yields the following pullback square:
    \[
      \begin{tikzcd}[column sep = 30pt]
        R \times GS
        \arrow{r}{\pi}
        \arrow{d}[swap]{R \times Gs}
        \pullbackangle{-45}
        &
        R \arrow{d}{i}
        \\
        R \times GX
        \arrow{r}{R \times Gf}
        &
        R \times GY
      \end{tikzcd}
    \]
    Note that since $G0 = 0$ the upper morphism is the left-hand
    product projection and the right-hand one the left-hand coproduct
    injection.

    Now suppose we have $g\colon Z \to GX$ with $Gf \cdot g = z$, where
    $z\colon Z \to 0 \to GY$ is the zero morphism. Then for the zero
    morphism $z'\colon Z \to 0 \to R$ we clearly have
    \[
      (R \times Gf)\cdot \pair{z',g} = \pair{z',z} =  i \cdot z',      
    \]
    since the latter two are both the zero morphism $Z \to R \times
    GY$. Therefore, there is a unique morphism $h\colon Z \to R \times GS$
    with $(R\times Gs) \cdot h = \pair{z',g}$ and $\pi\cdot h =
    z'$. This implies that $h = \pair{z', h'}$ for a unique morphism
    $h'\colon Z \to GS$ such that $Gs \cdot h' = g$, which proves the claim. 

    \takeout{
  \item We prove that $G$ preserves inverse images. Suppose that
    $f\colon X \to Y$ is a linear map and $s\colon S \monoto Y$ a
    subspace, and let $t\colon T \monoto X$ be the inverse image of
    $s$ under $f$. That means we have the pullback on the left below:
    \begin{equation}\label{diag:pbs}
      \begin{tikzcd}
        T
        \arrow{r}{f'}
        \arrow[>->]{d}[swap]{t}
        \pullbackangle{-45}
        &
        S
        \arrow[>->]{d}{s}
        \\
        X \arrow{r}{f}
        &
        Y
      \end{tikzcd}
      \qquad\qquad
      \begin{tikzcd}[column sep = 30pt]
        R \times GT
        \arrow{r}{R \times Gf'}
        \arrow[>->]{d}[swap]{R \times Gt}
        \pullbackangle{-45}
        &
        R\times GS
        \arrow[>->]{d}{R \times Gs}
        \\
        R\times GX \arrow{r}{R\times Gf}
        &
        R\times GY
      \end{tikzcd}
    \end{equation}
    Moreover, since $FX = R + GX$ preserves inverse images and $+$ is
    also product, the square on the right above is a pullback, too. It
    is our task to prove that its right-hand product component
    (obtained by omitting ``$R\, \times$'' everywhere) is a pullback,
    too.
    
    Suppose that we have linear maps $p\colon P \to GX$ and
    $q\colon P \to GS$ such that $Gf \cdot p = Gs \cdot q$. Then we
    clearly have that
    $(R \times Gf) \cdot (R \times p) = (R \times Gs) \cdot (R \times
    q)$, and so there exist a unique
    $h\colon R \times P \to R \times GT$ such that
    $(R\times Gt) \cdot h = (R \times p)$ and
    $(R\times Gf') \cdot h = (R\times q)$.  Then the following
    morphism
    \begin{equation}\label{eq:k}
      k = (P \xrightarrow{\langle z, \id\rangle}
      R \times P \xrightarrow{h}
      R \times GT\xrightarrow{\pi_r}
      GT),
    \end{equation}
    where $z\colon P \to R$ is the zero morphism and $\pi_r$ the
    right-hand product projection, is the desired unique morphism such
    that $Gt \cdot k = p$ and $Gf' \cdot k = q$. Indeed, the following
    diagram clearly commutes, which proves the two desired equations:
    \[
      \begin{tikzcd}[column sep = 30pt]
        &&
        P
        \arrow{d}[swap]{\langle 0, id\rangle}
        \arrow{rd}{q}
        \arrow{ld}[swap]{p}
        \\
        &
        GX
        \arrow{d}[swap]{\langle z, \id\rangle}
        \arrow[equals,bend right = 30]{ldd}
        &
        R \times P
        \arrow{d}[swap]{h}
        \arrow{rd}[near start]{R \times q}
        \arrow{ld}[near start, swap]{R \times p}
        &
        GS
        \arrow{d}{\langle z, \id\rangle}
        \arrow[equals,bend right = -30]{rdd}
        \\
        &
        R \times GX
        \arrow{ld}[swap]{\pi_r}
        &
        R \times GT
        \arrow{d}[swap]{\pi_r}
        \arrow{r}[swap]{R \times Gf'}
        \arrow{l}{R\times Gt}
        &
        R \times GS
        \arrow{rd}{\pi_r}
        \\
        GX
        &&
        GT
        \arrow{rr}{Gf'}
        \arrow{ll}[swap]{Gt}
        &&
        GS
      \end{tikzcd}
    \]
    For the uniqueness suppose we have $k'$ with $Gt \cdot k' = p$ and
    $Gf' \cdot k' = q$. Then $R\times k'$ is clearly a mediating
    morphism w.r.t.~the right-hand pullback in~\eqref{diag:pbs}, thus
    $h = R \times k'$. Using~\eqref{eq:k}, we see that
    \[
      k
      =
      \pi_r \cdot h \cdot \langle z,\id\rangle
      =
      \pi_r \cdot (R \times k') \cdot \langle z, \id\rangle
      =
      k' \cdot \pi_r \cdot \langle z, \id\rangle
      =
      k'.
    \]}
    
  \item Observe that $G0 = 0$, thus we can apply
    part~\ref{T:recwfvec:1}. Our recursive coalgebra
    $\alpha = \langle \alpha_1,\alpha_2\rangle\colon A \to R \times
    GA$ yields a coalgebra $\alpha_2\colon A \to GA$, and we prove
    that it is recursive, too. Indeed, given any algebra $\beta\colon
    GB \to B$, we have an algebra
    \[
      R \times GB \cong R + GB \xrightarrow{[z,\beta]} GB
    \]
    for $F$. Now observe that a morphism $h\colon A \to B$ is a
    coalgebra-to-algebra morphism for~$F$
    \[
      \begin{tikzcd}
        A
        \arrow{r}{\pair{\alpha_1,\alpha_2}}
        \arrow{d}[swap]{h}
        &
        R \times GA
        \arrow{d}{R \times Gh}
        \\
        B
        &
        R \times GB
        \arrow{l}[swap]{[0,\beta]}
      \end{tikzcd}
    \]
    iff it is a coalgebra-to-algebra morphism from $(A,\alpha_2)$ to
    $(B,\beta)$ for $G$. Since the former exists uniquely, so does the
    latter. This proves that $(A,\alpha_2)$ is recursive.

    By part~\ref{T:recwfvec:1} the coalgebra $(A,\alpha_2)$ is
    well-founded for $G$. Its next time operator $\nexttime$ is the
    same as that of the $F$-coalgebra $(A,\alpha)$ because in the
    diagram below the outside is a pullback iff the left-hand square
    is:
    \[
      \begin{tikzcd}
        \nexttime S
        \arrow[>->]{d}[swap]{\nexttime s}
        \arrow{r}{\alpha(s)}
        \pullbackangle{-45}
        &
        R \times GS
        \arrow[>->]{d}{R\times Gs}
        \arrow{r}{\pi_r}
        &
        GS
        \arrow[>->]{d}{Gs}
        \\
        A
        \arrow{r}{\alpha}
        \arrow[shiftarr = {yshift=-20pt}]{rr}{\alpha_2}
        &
        R \times GA
        \arrow{r}{\pi_r}
        &
        GA
      \end{tikzcd}
    \]    
%
%
    Since $\id_A$ is the unique fixed point of $\nexttime$ w.r.t.~$G$,
    it is also the unique fixed point w.r.t.~$F$. Thus $(A,\alpha)$ is
    well-founded for $F$ as desired. \qed
  \end{enumerate}
\end{proof}
\begin{corollary}\label{C:vec}
  For every functor on $\KVec$ preserving inverse images, the
  following properties of a coalgebra are equivalent:
  \[
    \mbox{well-foundedness} \iff
    \mbox{parametric recursiveness} \iff
    \mbox{recursiveness}.
  \]
\end{corollary}

\section{Conclusions}

\takeout{%
\lmnote{This needs to change. SM: Ok I changed it in the short version
and copied it here.}
We have presented well-founded coalgebras whose definition 
captures the concept of well-founded induction on an abstract
level. We have also provided a new proof of Taylor's General Recursion Theorem
stating that every well-founded coalgebra is parametrically
recursive. This holds for functors preserving monomorphisms on a
complete and well-founded category with \smooth monomorphisms. In
the category of sets, this even holds for every endofunctor, and the
converse holds for endofunctors preserving inverse images. Coming back
to our theme in this paper, for every set functor the initial
algebra is, equivalently, the terminal well-founded coalgebra.

Finally, we have also provided an iterative construction of the well-founded
part of a given coalgebra. It is carried by the least fixed
point of Jacobs' next time operator. In addition, the well-founded
part yields the coreflection of a coalgebra in the category of
well-founded coalgebras.
}

Well-founded coalgebras introduced by Taylor~\cite{taylor2} have a
compact definition based on an extension of Jacobs' `next time'
operator. Our main contribution is a new proof of Taylor's General
Recursion Theorem that every well-founded coalgebra is recursive,
generalizing this result to all endofunctors preserving monomorphisms
on a complete and well-powered category with \smooth
monomorphisms. For functors preserving inverse images, we also have
seen two variants of the converse implication ``recursive $\Rightarrow$
well-founded'', under additional hypothesis: one due to Taylor for
categories with a subobject classifier, and the second one provided
that the category has universally \smooth monomorphisms and the
functor has a pre-fixed point. Various counterexamples demonstrate
that all our hypotheses are necessary.


%
%

\bibliographystyle{splncs03}
\bibliography{refs}

\begin{thebibliography}{10}
\providecommand{\url}[1]{\texttt{#1}}
\providecommand{\urlprefix}{URL }

\bibitem{AdamekLuckeMilius07}
Ad{\'{a}}mek, J., L{\"{u}}cke, D., Milius, S.: Recursive coalgebras of finitary
  functors. Informatique Th\'{e}orique et Applications  41(4),  447--462
  (2007), \url{http://dx.doi.org/10.1051/ita:2007028}

\bibitem{A74}
Ad{\'{a}}mek, J.: Free algebras and automata realizations in the language of
  categories. Comment.~Math.~Univ.~Carolin.  15,  589--602 (1974)

\bibitem{ahs}
Ad{\'{a}}mek, J., Herrlich, H., Strecker, G.E.: Abstract and Concrete
  Categories: The Joy of Cats. Dover Publications, 3rd edn. (2009)

\bibitem{alm_rec}
Ad\'{a}mek, J., L\"ucke, D., Milius, S.: Recursive coalgebras of finitary
  functors. Theor.~Inform.~Appl.  41(4),  447--462 (2007)

\bibitem{amm18}
Ad\'amek, J., Milius, S., Moss, L.S.: Fixed points of functors.
  J.~Log.~Algebr.~Methods Program.  95,  41--81 (2018)

\bibitem{amms}
Ad\'amek, J., Milius, S., Moss, L.S., Sousa, L.: Well-pointed coalgebras.
  Log.~Methods Comput.~Sci.  9(2),  1--51 (2014)

\bibitem{amsw19}
Ad\'amek, J., Milius, S., Sousa, L., Wi\ss\/mann, T.: On finitary functors
  (2019), accepted for publication in {\em Theor.~Appl.~Categ.}; available
  online at \url{https://arxiv.org/abs/1902.05788}

\bibitem{ar}
Ad\'{a}mek, J., Rosick\'y, J.: Locally Presentable and Accessible Categories.
  Cambridge University Press (1994)

\bibitem{ATbook}
Ad\'{a}mek, J., Trnkov\'a, V.: Automata and Algebras in Categories, Mathematics
  and its Applications, vol.~37. Kluwer Academic Publishers (1990)

\bibitem{Borceux94}
Borceux, F.: Handbook of Categorical Algebra: Volume 1, Basic Category Theory.
  Encyclopedia of Mathematics and its Applications, Cambridge University Press
  (1994)

\bibitem{Borceux94_2}
Borceux, F.: Handbook of Categorical Algebra: Volume 2, Categories and
  Structures. Encyclopedia of Mathematics and its Applications, Cambridge
  University Press (1994)

\bibitem{cuv06}
Capretta, V., Uustalu, T., Vene, V.: Recursive coalgebras from comonads.
  Inform.~and Comput.  204,  437--468 (2006)

\bibitem{cuv09}
Capretta, V., Uustalu, T., Vene, V.: Corecursive algebras: A study of general
  structured corecursion. In: Oliveira, M., Woodcock, J. (eds.) Formal Methods:
  Foundations and Applications, Lecture Notes in Computer Science, vol. 5902,
  pp. 84--100. Springer Berlin Heidelberg (2009)

\bibitem{eppendahl99}
Eppendahl, A.: Coalgebra-to-algebra morphisms. In: Proc.~Category Theory and
  Computer Science (CTCS). Electron.~Notes Theor.~Comput.~Sci., vol.~29, pp.
  42--49 (1999)

\bibitem{Freyd64}
Freyd, P.J.: Abelian Categories: An Introduction to the Theory of Functors.
  Harper and Row (1964)

\bibitem{Gumm2005}
Gumm, H.: From {$T$}-coalgebras to filter structures and transition systems.
  In: Fiadeiro, J.L., Harman, N., Roggenbach, M., Rutten, J. (eds.) Algebra and
  Coalgebra in Computer Science, Lecture Notes in Computer Science, vol. 3629,
  pp. 194--212. Springer Berlin Heidelberg (2005)

\bibitem{Jacobs02}
Jacobs, B.: The temporal logic of coalgebras via {G}alois algebras.
  Math.~Structures Comput.~Sci.  12(6),  875--903 (2002)

\bibitem{JeanninEA17}
Jeannin, J.B., Kozen, D., Silva, A.: Well-founded coalgebras, revisited.
  Math.~Structures Comput.~Sci.  27,  1111--1131 (2017)

\bibitem{Kurz00}
Kurz, A.: Logics for Coalgebras and Applications to Computer Science. Ph.D.
  thesis, Ludwig-Maximilians-Universität München (2000)

\bibitem{lambek}
Lambek, J.: A fixpoint theorem for complete categories. Math.~Z.  103,
  151--161 (1968)

\bibitem{Lawvere70}
Lawvere, W.F.: Quantifiers and sheaves. Actes Cong\`es Intern.\ Math.  1,
  329--334 (1970)

\bibitem{MP92}
Manna, Z., Pn\"ueli, A.: The Temporal Logic of Reactive and Concurrent Systems:
  Specification. Springer-Verlag (1992)

\bibitem{MG82}
Meseguer, J., Goguen, J.A.: Initiality, induction, and computability. In:
  Algebraic methods in semantics ({F}ontainebleau, 1982), pp. 459--541.
  Cambridge Univ. Press, Cambridge (1985)

\bibitem{milius}
Milius, S.: Completely iterative algebras and completely iterative monads.
  Inform.~and Comput.  196,  1--41 (2005)

\bibitem{mpw19}
Milius, S., Pattinson, D., Wi\ss\/mann, T.: A new foundation for finitary
  corecursion and iterative algebras. Inform.~and Comput.  (2019), {T}o appear;
  available online at \url{https://doi.org/10.1016/j.ic.2019.104456}.

\bibitem{osius}
Osius, G.: Categorical set theory: a characterization of the category of sets.
  J. Pure Appl. Algebra  4(79--119) (1974)

\bibitem{taylor3}
Taylor, P.: Towards a unified treatment of induction~{I}: the general recursion
  theorem (1995--6), preprint, available at
  \url{www.paultaylor.eu/ordinals/\#towuti}

\bibitem{taylor2}
Taylor, P.: Practical Foundations of Mathematics. Cambridge University Press
  (1999)

\bibitem{takr}
Trnkov\'{a}, V., Ad\'{a}mek, J., Koubek, V., Reiterman, J.: Free algebras,
  input processes and free monads. Comment. Math. Univ. Carolin.  16,  339--351
  (1975)

\bibitem{trnkova69}
Trnkov\'a, V.: Some properties of set functors. Comment.~Math.~Univ.~Carolin.
  10,  323--352 (1969)

\bibitem{trnkova71}
Trnkov\'a, V.: On a descriptive classification of set functors {I}.
  Comment.~Math.~Univ.~Carolin.  12,  143--174 (1971)

\bibitem{WissmannEA19}
Wi\ss{}mann, T., Milius, S., ya~Katsumata, S., Dubut, J.: A coalgebraic view on
  reachability, submitted; available online at
  \url{https://arxiv.org/abs/1901.10717}

\end{thebibliography}

\end{document}
